\documentclass[11pt, letter]{article}

\usepackage[table,svgnames]{xcolor}
\usepackage{hhline}
\usepackage[colorlinks=true,linkcolor=black,citecolor=MidnightBlue,urlcolor=MidnightBlue]{hyperref}
\usepackage[letterpaper,margin=1in]{geometry}

\usepackage{amssymb,amsmath,amsfonts,fullpage,mdwlist,enumerate,amsthm, graphicx,algpseudocode, algorithm, algorithmicx, cite, framed, thm-restate}

\newtheorem{theorem}{Theorem}[section]
\newtheorem{lemma}[theorem]{Lemma}
\newtheorem{observation}{Observation}

\newtheorem{claim}[theorem]{Claim}
\newtheorem{definition}{Definition}[section]

\usepackage[textsize=tiny]{todonotes}

\newcommand{\remove}[1]{}

\newcommand{\eps}{\varepsilon}

\DeclareMathOperator{\poly}{poly}

\makeatletter
\renewcommand{\paragraph}{%
  \@startsection{paragraph}{4}%
  {\z@}{2.5ex \@plus 1ex \@minus .2ex}{-1em}%
  {\normalfont\normalsize\bfseries}%
}
\makeatother

\title{Improved Distributed Approximations for Minimum-Weight Two-Edge-Connected Spanning Subgraph}

\author{Michal Dory \\
\small Technion, Israel \\
\small smichald@cs.technion.ac.il 
\and
Mohsen Ghaffari \\
\small ETH Zurich, Switzerland \\
\small ghaffari@inf.ethz.ch
}

\begin{document}

\begin{titlepage}

\date{}
\maketitle


\begin{abstract} 
The minimum-weight $2$-edge-connected spanning subgraph (2-ECSS) problem is a natural generalization of the well-studied minimum-weight spanning tree (MST) problem, and it has received considerable attention in the area of network design. The latter problem asks for a minimum-weight subgraph with an edge connectivity of $1$ between each pair of vertices while the former strengthens this edge-connectivity requirement to $2$. Despite this resemblance, the 2-ECSS problem is considerably more complex than MST. While MST admits a linear-time centralized exact algorithm, 2-ECSS is NP-hard and the best known centralized approximation algorithm for it (that runs in polynomial time) gives a $2$-approximation.

In this paper, we give a \emph{deterministic} distributed algorithm with round complexity of $\widetilde{O}(D+\sqrt{n})$ that computes a ($5+\eps$)-approximation of 2-ECSS, for any constant $\eps>0$. Up to logarithmic factors, this complexity matches the $\widetilde{\Omega}(D+\sqrt{n})$ lower bound that can be derived from the technique of Das Sarma et al. [STOC'11], as shown by Censor-Hillel and Dory [OPODIS'17]. Our result is the first distributed \emph{constant} approximation for 2-ECSS in the nearly optimal time and it improves on a recent randomized algorithm of Dory [PODC'18], which achieved an $O(\log n)$-approximation in $\widetilde{O}(D+\sqrt{n})$ rounds. Our algorithm also gives a deterministic $\widetilde{O}(D+\sqrt{n})$-round $(4+\eps)$-approximation for the closely related weighted tree augmentation problem (TAP), a similar result was known before only for the \emph{unweighted} variant of the problem.

We also present an alternative algorithm for $O(\log n)$-approximation, whose round complexity is linear in the \emph{low-congestion shortcut} parameter of the network---following a framework introduced by Ghaffari and Haeupler [SODA'16]. This algorithm has round complexity  $\widetilde{O}(D+\sqrt{n})$ in worst-case networks but it provably runs much faster in many well-behaved graph families of interest. For instance, it runs in $\widetilde{O}(D)$ time in planar networks and those with bounded genus, bounded path-width or bounded tree-width.

\end{abstract}

\thispagestyle{empty}
\end{titlepage}

\section{Introduction and Related Work}
There is a large body of deep and beautiful work on developing centralized (approximation) algorithms for various \emph{network design} problems. These are typically in the format of finding a minimum-weight subgraph $H$ of a graph/network $G$, given certain weights for the network elements (edges or vertices), such that $H$ satisfies some desired properties, e.g., certain connectivity requirements. The motivation is simple and practical: suppose that using each communication link of $G$ has some cost (e.g., in monetary terms, to be paid to the provider). Which links should we use to minimize our costs while still meeting all of our connectivity requirements? Unlike the extensive amount of attention and progress on centralized algorithms for network design problems, developments on the distributed side have been much more scarce. Even for many of the basic problems of the area, we do not have satisfactory distributed algorithms. There are only few exceptions.

One prominent exception is the case of the minimum-weight spanning tree (MST) problem. This is the minimum-weight subgraph that ensures that all vertices are connected to each other. MST has been one of the central problems in the developments of distributed graph algorithms, starting from the pioneering work of Gallager et al.\cite{gallager1983distributed}. By now we have a relatively good understanding of the complexity of this problem: It can be solved in $O(D+\sqrt{n}\log^* n)$  rounds \cite{kutten1998fast}, in synchronous message-passing rounds with $O(\log n)$-bit messages (i.e., the $\mathsf{CONGEST}$ model\cite{peleg:2000}), where $D$ denotes the network diameter and $n$ denotes the number of vertices. Moreover, this bound almost matches an $\Omega(D+\sqrt{\frac{n}{\log n}})$ lower bound, developed in the line of work of \cite{peleg2000near, elkin2006unconditional, das2011distributed}.

As soon as we take just one step beyond the basic MST problem, in the realm of network design questions, we soon reach a problem for which we do not have a satisfactory algorithmic answer, yet. This is the minimum-weight $2$-edge-connected spanning subgraph (2-ECSS) problem, where instead of asking all vertices to be connected, we ask them to be $2$-edge-connected. That is, we want that the subgraph has some minimal resilience to edge failures and even if one of the links breaks, still all vertices are connected (without us having to recompute the subgraph). This is again a natural and practical property to desire in networking, given the prevalence and frequency of link failures. Although there has been recent progress on the distributed complexity of this problem, the answers are still not satisfactory, and that is where the contribution of this paper comes in. But before delving into our contribution, let us take a closer look at what is known about 2-ECSS.

The problem of computing a minimum 2-ECSS is known to be NP-hard, by a reduction from the Hamiltonian cycle problem~\cite{cheriyan1998improved}. This already shows that 2-ECSS is much harder than the closely related problem of MST, which admits a linear-time (randomized) centralized exact algorithm~\cite{karger1995randomized}. The best known centralized (polynomial-time) approximation algorithms for 2-ECSS give a $2$-approximation~\cite{khuller1994biconnectivity, jain2001factor}. These algorithms also give a $2$-approximation for the more general minimum-weight $k$-edge-connected spanning subgraph ($k$-ECSS).
If all edge-weights are equal (i.e., in the unweighted case), then the best known approximation factors become $4/3$ and $1+\frac{1}{2k}+O(\frac{1}{k^2})$, respectively, for 2-ECSS \cite{DBLP:journals/combinatorica/SeboV14} and $k$-ECSS \cite{gabow2012iterated}.

\subsection{Our Contributions, in the Context of the Distributed State-of-the-Art}
\subsubsection{First Contribution --- Better Approximation}

\paragraph{Distributed state-of-the-art:} Censor-Hillel and Dory\cite{DBLP:conf/opodis/Censor-HillelD17} provided an algorithm that runs in $O(h_{MST}+\sqrt{n}\log^* n)$ rounds and gives a $3$-approximation of the min-weight $2$-ECSS. Here, $h_{MST}$ denotes the height of the minimum spanning tree; note that this can be as large as $\Theta(n)$, even in networks with small diameter $D$. Thus, this round complexity can be much higher than the lower bound of $\widetilde{\Omega}(D+\sqrt{n})$\cite{DBLP:conf/opodis/Censor-HillelD17, das2011distributed}. Dory\cite{DBLP:conf/podc/Dory18} provided a randomized algorithm that runs in the near-optimal run time of $\widetilde{O}(D+\sqrt{n})$ and provides an $O(\log n)$-approximation.

Two comments are in order. First, for general $k$, one can obtain an $O(\log k)$-approximation but in round complexity of $O(knD)$~\cite{shadeh2009distributed, goemans1994improved}, and an $O(k\log n)$-approximation in $O(k(D\log^3 n+n))$ rounds~\cite{DBLP:conf/podc/Dory18}. Both of these complexities are at least linear in the network size and well-above our target round complexity. 
Second, let us discuss the special case where all the weights are equal, i.e., when the problem is unweighted. This case is much easier and one can obtain a constant approximation locally in $(k\log^{1+o(1)} n)$ rounds \cite{kECSS}. 
Additionally, if we are interested in better approximations for \emph{unweighted} 2-ECSS, there is an $O(D)$-round $2$-approximation \cite{DBLP:conf/opodis/Censor-HillelD17} and an $O(n)$-round $3/2$-approximation \cite{krumke2007distributed}. 

\paragraph{Our First Result:} Above we mentioned a $3$-approximation with suboptimal time~\cite{DBLP:conf/opodis/Censor-HillelD17} and an optimal time with suboptimal approximation of $O(\log n)$~\cite{DBLP:conf/podc/Dory18}.
Our first main result is to provide an algorithm that gets the best of these two worlds (using a different technique), namely near-optimal time and constant approximation. More concretely, we show the following result:

\begin{restatable}{theorem}{twoECSS}
\label{2ECSS-thm}
There is a deterministic $(5+\epsilon)$-approximation algorithm for weighted 2-ECSS in the $\mathsf{CONGEST}$ model that takes $O((D+\sqrt{n})\frac{\log^2{n}}{\epsilon})$ rounds.
\end{restatable}

\remove{
\begin{theorem}
There is a deterministic $(9+\epsilon)$-approximation algorithm for weighted 2-ECSS in the CONGEST model that takes $O((D+\sqrt{n})\frac{\log^2{n}}{\epsilon})$ rounds.
\end{theorem}  

\begin{theorem} There is a deterministic distributed algorithm that in $O((D+\sqrt{n})\log^2{n})$ rounds, computes a $5+\eps$ approximation of weighted $2$-ECSS, for any constant $\eps>0$.
\end{theorem}
}

We think that achieving a constant approximation in the near-optimal time is the main strength of this result. Another advantage of our approach is that is gives a \emph{deterministic} algorithm, while the $O(\log n)$-approximation algorithm of \cite{DBLP:conf/podc/Dory18} was randomized.
Our algorithm also gives an $O((D+\sqrt{n})\log^2{n})$-round  $(4+\epsilon)$-approximation for the closely related weighted tree augmentation problem (TAP), a similar result was known before only for the \emph{unweighed} variant of the problem \cite{DBLP:conf/opodis/Censor-HillelD17}.


\subsubsection{Second Contribution --- Faster Algorithm in Well-Behaved Networks}
Our second contribution is in a different direction. Instead of setting the $\widetilde{\Omega}(D+\sqrt{n})$ lower-bound as our target round complexity for all network topologies, we would like to have an algorithm that runs no slower than this in the worst-case, but also ideally much faster in less pathological network instances, e.g., in graph families which are more relevant in practical networking settings.

In this regard, we build on a framework set forth by Ghaffari and Haeupler\cite{ghaffari2016shortcuts}, known as \emph{low-congestion shortcuts}. This framework was first used to give $\widetilde{O}(D)$ round distributed algorithms for problems such as minimum spanning tree and minimum cut, in planar networks. Notice that this can be much faster than the $\widetilde{O}(D+\sqrt{n})$ bound. By now, it is known that the same concept can be applied to a wider range of graph families. We mention just a sampling here: in networks with bounded genus, bounded path-width, and bounded tree-width, one can also obtain an $\widetilde{O}(D)$ round MST algorithm\cite{ghaffari2016shortcuts, haeupler2016low}. In the much more general family of excluded minor graphs, one can obtain an $\widetilde{O}(D^2)$ round MST algorithm\cite{haeupler2018minor}. Finally, in Erdos-Renyi random graphs $G_{n, p}$ with $p=\Omega(\log n/n)$, one can obtain an MST algorithm with complexity $2^{O(\sqrt{\log n})}$\cite{ghaffari2017distributed, ghaffari2018new}. 

These different round complexities come from the quality of the \emph{shortcuts} that can be obtained in the given graph family, and the time needed to construct them. Given a graph $G=(V, E)$, we say it admits $\alpha$-congestion $\beta$-dilation shortcuts with construction time $\gamma$, if the following is satisfied: for any partitioning of $V$ into vertex-disjoint parts $V_1, V_2, \dots, V_N$, each of which induced a connected subgraph $G[V_i]$, in $\gamma$ rounds, we can construct subgraphs $H_1$, \dots $H_N$ such that (A) $\forall i\in\{1,2, \dots, N\}$, the subgraph $G[V_i]+H_i$ has diameter at most $\beta$ and moreover, (B) each edge $e\in E$ appears in at most $\alpha$ many of the graphs $G[V_i]+H_i$. It was shown by Ghaffari and Haeupler\cite{ghaffari2016shortcuts} that MST and $(1+\eps)$-approximation of minimum cut can be solved in $\widetilde{O}(\alpha+\beta+\gamma)$ time. It is notable that this \emph{shortcut complexity} $\mathsf{SC}(G)=\alpha+\beta+\gamma$ is in the worst-case $O(D+\sqrt{n})$, as shown by Ghaffari and Haeupler\cite{ghaffari2016shortcuts}, but it can be much better: for instance, as listed above, we have $\mathsf{SC}(G) = \widetilde{O}(D)$ for any network that is planar, bounded genus, bounded path-width, or bounded tree-width\cite{ghaffari2016shortcuts, haeupler2016low}. Moreover, we have $\mathsf{SC}(G) = \widetilde{O}(D^2)$ in any excluded minor graph\cite{haeupler2018minor}, and $\mathsf{SC}(G) = 2^{O(\sqrt{\log n})} \ll n^{o(1)}$ in any Erdos-Renyi graphs above the connectivity threshold\cite{ghaffari2017distributed, ghaffari2018new}. Hence, to summarize, the framework of shortcuts and the parameter $\mathsf{SC}(G)=\alpha+\beta+\gamma$ provide a convenient way to develop algorithms that are in the worst-case as bad as the lower bound $\widetilde{\Omega}(D+\sqrt{n})$, but can also run much faster, in a provable way, in many graph families of interest.

\paragraph{Our Second Result:} Our second contribution is a distributed algorithm with such a graceful complexity for $O(\log n)$-approximation of the minimum weight 2-ECSS problem.
 
\begin{theorem}\label{thm:shortcut} There is a randomized distributed algorithm that computes an $O(\log n)$-approximation of weighted $2$-ECSS in the $\mathsf{CONGEST}$ model and runs in $\widetilde{O}(\mathsf{SC}+D)=\widetilde{O}(\alpha+\beta+\gamma+D)$ rounds, in any network $G$ with diameter $D$ where for any partition we can build an $\alpha$-congestion $\beta$-dilation shortcut, in $\gamma$ rounds. 
\end{theorem}
Theorem~\ref{thm:shortcut} gives an algorithm that always runs in at most $\widetilde{O}(D+\sqrt{n})$ rounds, but it runs provably much faster in many graph families of interest. In particular, it implies that we can obtain an $O(\log n)$-approximation in $\widetilde{O}(D)$ rounds in any network that is planar, bounded genus, bounded path-width, or bounded tree-width, in $\widetilde{O}(D^2)$ rounds in any network that is excluded minor, and in $2^{O(\sqrt{\log n})}$ rounds in Erdos-Renyi random graphs $G_{n, p}$ with $p=\Omega(\log n/n).$





\subsection{Technical overview} \label{sec:tech_overview}

\paragraph{2-ECSS Approximation Boils Down to Tree Augmentation:} 
Let $T$ be the MST of $G=(V,E)$. To obtain the minimum weight $2$-ECSS, we would like to compute a minimum weight set of non-tree edges $A \subseteq E$ that \emph{covers} all the tree edges of $T$, according to the following definition: We say that a non-tree edge $e \in E$ covers a tree edge $t \in T$ if $(T \cup \{e\}) \setminus \{t\}$ is connected. We say that a set of non-tree edges $A$ covers the tree edge $t \in T$ if $A$ contains an edge that covers $t$. Note that $A$ covers all the tree edges iff $T \cup A$ is 2-edge-connected. The problem of computing the minimum cost set of edges $A$ such that $T \cup A$ is 2-edge-connected is known as the \emph{tree augmentation problem} (TAP) and it can be seen that any $\alpha$-approximation for it, along with the MST, $T$, gives an $\alpha+1$ approximation of 2-ECSS (See Section \ref{sec:pre} for a formal statement). This follows from the following. First, the MST is no more costly than the optimal 2-ECSS, as the latter contains a spanning tree in it. Second, the optimal 2-ECSS also covers all the edges of the MST and thus its size is an upper bound on the size of the optimal tree augmentation. Therefore, our objective boils down to providing an approximation for TAP.

\paragraph{Tree Augmentation Can be Seen as a Set Cover Problem:} TAP is a special case of the well-studied \emph{minimum-cost set cover} problem, where we want to choose a minimum-cost collection of \emph{sets} which cover all the \emph{elements}\cite{vazirani2013approximation}; in our case, this is choosing a minimum-cost collection of non-tree edges $A \subseteq E$ which cover all the tree edges. Of course, in the distributed setting, this special case comes with its own difficulties, because the sets (in our case, non-tree edges) are not directly connected to their elements (in our case, tree edges), which means we cannot resort to standard distributed set cover algorithms in a naive way. However, distributed and parallel algorithms for set cover can still form a basis for a distributed algorithm for TAP.
\paragraph{Previous Approaches and the Challenge:} 
The connection between TAP and set cover is used in \cite{DBLP:conf/podc/Dory18}, where it is shown that using a certain decomposition of the graph it is possible to solve TAP by simulating a distributed greedy algorithm for set cover. This results in a randomized $O((D+\sqrt{n})\log^2{n})$-round $O(\log{n})$-approximation for weighted TAP and weighted 2-ECSS. However, since the approximation ratio of the greedy algorithm for set cover is $O(\log{n})$, we cannot obtain a better approximation by simulating a general algorithm for set cover. Another option is to design a specific algorithm that exploits the properties of TAP. This is done in \cite{DBLP:conf/opodis/Censor-HillelD17}, where it is shown how to obtain a 2-approximation for weighted TAP and a 3-approximation for weighted 2-ECSS in $O(h)$ rounds, where $h$ is the height of the tree we augment. However, in the worst case, $h$ can be linear, and the algorithm from \cite{DBLP:conf/opodis/Censor-HillelD17} relies heavily on the fact that we need to scan the whole tree.

\paragraph{Our Approach:}
To overcome the above obstacles, we suggest a new algorithm for weighted TAP. As in \cite{DBLP:conf/podc/Dory18}, our algorithm simulates an algorithm for set cover, however, we simulate an algorithm that exploits the specific structure of the set cover problem we solve, which allows to get a \emph{constant} approximation. A crucial ingredient in the algorithm is a certain layering of the graph, that allows to implement the algorithm efficiently, and results in a \emph{deterministic} algorithm that is near-optimal both in terms of the round complexity and in the approximation ratio obtained. 
In more detail, our algorithm is inspired by a parallel algorithm for set cover instances with \emph{small neighbourhood covers} \cite{agarwal2018set}, and we next give the high-level intuition of this property in our case. As a first step, we start by replacing the input graph $G$ by a related virtual graph $G'$, having the additional property that all the non-tree edges in $G'$ are between ancestors to descendants. This is useful, because now the set cover problem we solve has an additional structure. Consider a path $P$ with a root $r$, from all the non-tree edges that cover a tree edge $t \in P$ we can choose only \emph{two} non-tree edges that cover exactly the same tree edges in $P$. This is done by taking the edges that cover $t$ and go to the highest ancestor, or to the lowest descendant. 
If our tree is not a path, we can decompose it in a certain way to $O(\log{n})$ layers, each one is composed of disjoint paths, and a similar property now holds for all the tree edges with respect to these paths. 
A generalization of this property called the \emph{small neighbourhood cover} property is studied in \cite{agarwal2018set}, where the authors show an efficient parallel algorithm that obtains a constant approximation for set cover instances having this property. Our general approach is to simulate the algorithm from \cite{agarwal2018set} on the virtual graph $G'$. A major difference in our setting is that the set cover graph is not given as an input, and in particular tree edges cannot communicate directly with non-tree edges that cover them. We next explain how we overcome this.

\paragraph{Distributed Implementation:} 
The algorithm from \cite{agarwal2018set} is a primal-dual algorithm that is composed of two main phases, the forward phase and the reverse-delete phase. For implementing the forward phase, we show that the only communication pattern between tree edges and non-tree edges in the algorithm is computation of \emph{aggregate} functions. To implement such computations, we bring to our construction a decomposition of the tree into segments used in \cite{DBLP:conf/podc/Dory18, ghaffari2016near}. While similar computations are done in \cite{DBLP:conf/podc/Dory18}, the main obstacle in our case is that our algorithm works with a virtual graph $G'$ and we explain how to extend the algorithm to this case. The reverse-delete phase creates new obstacles, as it requires computation of a maximal independent set (MIS) in a \emph{completely virtual} graph $G_i$.
In $G_i$, the vertices are some of the tree edges, and two tree edges are connected by an edge if there is a non-tree edge that covers them. We present a specific algorithm to solve this task, which relies heavily on the structure of the decomposition and the layering. The high-level idea is first to compute an MIS of an induced subgraph of $G_i$ that has only a carefully chosen set of $O(\sqrt{n})$ tree edges. We show that it is enough to let vertices learn about these edges and a constant number of non-tree edges that cover each of them, to simulate an MIS computation on the induced subgraph. Then, we let all the vertices simulate a local algorithm inside each segment, to add additional uncovered tree edges to the MIS. Finally, we prove that although we work on different segments at the same time, all the edges added to the MIS by different segments really form an MIS in the virtual graph $G_i$. Based on these ideas, we show that we can implement the whole algorithm in $\widetilde{O}(D+\sqrt{n})$ rounds.
This algorithm obtains a $(4+\epsilon)$-approximation for weighted TAP in the virtual graph $G'$. This translates to a $(8+\epsilon)$-approximation for weighted TAP in the original graph $G$, which results in a $(9+\epsilon)$-approximation for weighted 2-ECSS.

\paragraph{Improved Approximation:} To obtain an improved approximation of $(4+\epsilon)$ for weighted TAP and $(5+\epsilon)$ for weighted 2-ECSS, we change some elements in the algorithm. At a high-level, the analysis that gives a $(4+\epsilon)$-approximation for weighted TAP in the virtual graph $G'$ is based on showing that we cover certain edges in the tree at most $4$ times, by changing several elements in the algorithm and using a careful case analysis we can get an algorithm that covers these edges only 2 times, which improves the approximation of the whole algorithm.

\paragraph{A bird's eye view of our second algorithm:}
Our second algorithm which proves Theorem \ref{thm:shortcut} and obtains an $O(\log n)$-approximation in time proportional to the shortcut complexity of the network, appears in Section \ref{app:secondAlgo}. Here, we provide a very brief summary and comment on the novel components.

On the outer layer, this second algorithm also works by first computing an MST, $T$, and then finding an approximately minimum-weight augmentation for it. Also, we again view this tree augmentation as a set cover problem, where we would like to find a minimum-cost collection of non-tree edges, which cover all the tree edges of $T$. We follow a standard parallelization of the sequential greedy algorithm for set cover \cite{berger1989efficient}, which gives an $O(\log n)$-approximation. We gradually add more and more of the most cost-effective non-tree edges to the solution, while ensuring that they cover mostly different tree edges. Cost-effectiveness refers to how many tree edges we can cover, per unit of weight. We refer the reader to Section \ref{app:secondAlgo} for more detailed outline, and also to \cite{DBLP:conf/podc/Dory18}. %
The main novelty in our algorithm is in how to perform each iteration of the parallel set cover, in time proportional to the shortcut complexity of the graph. Concretely, we will need two subroutines: (A) Given a collection of non-tree edges, any tree edge should know whether it is covered by this collection or not. (B) Each non-tree edge should know the number of tree edges that it can cover. To solve these two subroutines, we build two algorithmic tools, in the framework of shortcuts\cite{ghaffari2016shortcuts}. A tool that allows each node to know the summation of the values of all of its ancestors, in an arbitrary given tree, and a tool that allows us to compute the \emph{heavy-light} decomposition of any tree. These would be trivial to do in time proportional to the height of the tree, but we want the complexity to be proportional to the shortcut complexity of the network. We present these tools in Section \ref{subsec:log-tools}. Section \ref{subsec:log-subroutines} explains how we use these tools to build the subroutines (A) and (B) mentioned above. The problems solved by our tools are very basic and frequently used computations about a tree. Thus, we hope that these tools will find applications beyond our work.

\remove{ 
We show that we can implement the algorithm in $\widetilde{O}(D+\sqrt{n})$ rounds and obtain a $(4+\epsilon)$-approximation for weighted TAP in the virtual graph $G'$. This gives a $(8+\epsilon)$-approximation for weighted TAP in the original graph $G$. By Claim \ref{claim_TAP_2ECSS}, this results in a $(9+\epsilon)$-approximation for weighted 2-ECSS. We next provide an overview of the algorithm. Full proofs and implementation details appear in Section \ref{sec:imp}. 
}

\subsection{Additional Related Work} 
\label{app:OtherRelatedWork}

\paragraph{Small Neighborhood Covers:} Our algorithm is inspired by a parallel algorithm for set cover problems with the small neighborhood cover property (SNC) \cite{agarwal2018set}. This class includes central problems such as vertex cover, interval cover and bag cover. 
The approximation obtained depends on a parameter $\tau$ that is related to the size of the neighborhood cover, and the time complexity depends on a certain layering of the problem. 

\textbf{The SNC property.} We next give an intuition for the SNC property, for a formal definition see \cite{agarwal2018set}. In a set cover problem, the goal is to cover a universe of elements by a minimum cost collection of sets. Two elements in the universe are \emph{neighbors} if there is a set that covers both of them, and the \emph{neigborhood} of an element includes all its neighbors. The $\tau$-SNC property says that each collection of sets that cover some element $u$ can be replaced by only $\tau$ sets that cover $u$ and all its neighbors from certain layers.
In our case, where we want to solve weighted TAP on the virtual graph $G'$, the parameter $\tau=2$, the elements are tree edges, and the sets correpond to non-tree edges. Two tree edges are neighbors if there is a non-tree edge that covers both of them. The SNC property is related to the following: if you look at a tree edge $t$ and a set of non-tree edges that cover it, you can replace them by only $\tau=2$ edges that cover $t$ and all its neighbours that are at the layer of $t$ or above it. 

\textbf{The approximation obtained.} In \cite{agarwal2018set}, they show general sequential and parallel algorithms that obtain $\tau$ and $(2+\epsilon)\tau^2$ approximations respectively for set cover problems with the $\tau$-SNC property, and leave as an open question whether an efficient parallel algorithm can obtain a $\tau$-approximation as the sequential one. Following the algorithm from \cite{agarwal2018set}, allows obtaining a $(2+\epsilon)\tau^2=(2+\epsilon)4$-approximation for weighted TAP in the virtual graph $G'$. Changing slightly the forward phase leads to a $(1+\epsilon)4$-approximation. Our improved approximation algorithm gives an approximation of $(\tau+\epsilon)=(2+\epsilon)$ for the same problem, almost matching the approximation obtained by the sequential algorithm in \cite{agarwal2018set}. Also, as we explain in Section \ref{sec:unweighted_approx}, for unweighted problems it seems that a simple variant of the algorithm from \cite{agarwal2018set} can actually give a $\tau$-approximation.\\

There are a number of other problems which have been studied in distributed graph algorithms, and have some resemblance to 2-ECSS. To the best of our understanding, none of these are too directly related to the problem that we tackle in this paper or to our method. 

\paragraph{Minimum Cut:} One problem is that, given a network $G$, we wish to compute or approximate the edge-connectivity of $G$ or, in the more general case of edge-weighted graphs, the minimum cut size of the graph. 
Ghaffari and Kuhn\cite{ghaffari2013cut} gave a $2+\eps$ approximation algorithm that runs in $\widetilde{O}(D+\sqrt{n})$ rounds, for any constant $\eps>0$. Nanongkai and Su \cite{nanongkai2014almost} gave a $1+\eps$ approximation algorithm that runs in $\widetilde{O}(D+\sqrt{n})$ rounds, and an exact algorithm for edge connectivity with round complexity of $\widetilde{O}((D+\sqrt{n})k^4)$, in graphs with edge-connectivity $k$. As mentioned above, a work of Ghafffari and Haeupler shows how to obtain a $(1+\eps)$ approximation in $\widetilde{O}(\mathsf{SC}(G))$ time, where $\mathsf{SC}(G)$ denotes the shortcut complexity of the graph. Finally, very recently, Daga et al.\cite{Daga2019cuts} have shown the first sublinear time algorithm for edge-connectivity for general $k$.

\paragraph{Fault-Tolerant MST:} Another related problem is that of computing a fault-tolerant minimum spanning tree, defined as follows: for each edge $e$ in the minimum spanning tree $T$ of $G$, we should know a minimum weight edge $e'$ such that $(T\cup \{e'\})\setminus \{e\}$ is a minimum spanning tree of $G\setminus \{e\}$. We note that despite the partial resemblance in the motivations, we do not see any real connection between this problem and 2-ECSS and an fault-tolerant MST may be far more heavy in weight than the minimum-weight $2$-edge connected subgraph; in the former the goal is to have an MST in the graph remaining after the failure, in the latter the goal is to have a minimum-weight graph that can survive a link failure. Ghaffari and Parter\cite{ghaffari2016near} showed an algorithm that computes fault-tolerant minimum spanning tree sturcutres, in the above sense, in $\widetilde{O}(D+\sqrt{n})$ rounds.

\vspace{-10pt}
\section{Preliminaries} \label{sec:pre}
\vspace{-5pt}
\paragraph{Model:} Throughout, we work with the standard $\mathsf{CONGEST}$ model of distributed computing\cite{peleg:2000}: The network is abstracted as an undirected graph $G=(V, E)$ and communication happens in synchronous rounds. Per round, each vertex can send $O(\log n)$ bits to each of its neighbors. In the distributed setting the input and output are local. At the beginning, each vertex knows the ids of its neighbours and the weights of the edges adjacent to it, and at the end, each vertex should know only a local part of the output. For example, when we compute a 2-edge-connected subgraph, each vertex should know which edges adjacent to it are taken to the solution.

\paragraph{Definitions:} We say that an undirected graph $G$ is \emph{2-edge-connected} if it remains connected after the removal of any single edge.
In the minimum weight 2-edge-connected spanning subgraph problem (2-ECSS), the goal is to find the minimum weight spanning subgraph that is 2-edge-connected. 
Given a tree $T$, we denote by $P_{u,v}$ the unique tree path between $u$ and $v$, and by $p(v)$ the parent of $v$ in the tree.
It is easy to see that $e=\{u,v\}$ covers a tree edge $t$ iff $t \in P_{u,v}.$ 

As explained in Section \ref{sec:tech_overview}, a natural approach for solving 2-ECSS is to start by computing a minimum spanning tree (MST), and then add to it a minimum weight set of edges that augments its connectivity to 2. This motivates the \emph{tree augmentation problem (TAP)}. In TAP, the input is a 2-edge-connected graph $G$ and a spanning tree $T$, and the goal is to find a minimum weight set of edges $A$ such that $T \cup A$ is 2-edge-connected. 
The following claim shows the connection between TAP and 2-ECSS.

\begin{claim} \label{claim_TAP_2ECSS}
An $\alpha$-approximation algorithm $Alg$ for TAP gives an $(\alpha+1)$-approximation algorithm for 2-ECSS, with complexity $O(T_n+D+\sqrt{n}\log^{*}{n})$ rounds, where $T_n$ is the complexity of $Alg$.
\end{claim}

\begin{proof}
Computing an MST, $T$, and then augmenting its connectivity using the approximation algorithm for TAP results in a solution of weight $w(T)+\alpha\cdot w(A^*)$, where $w(T)$ is the weight of the MST, and $w(A^*)$ is the weight of an optimal augmentation for $T$. An optimal solution $OPT$ for 2-ECSS is clearly of weight at least $w(T)$ and at least $w(A^*)$, since $OPT$ is a valid augmentation. Hence, the approximation obtained for 2-ECSS is at most $(\alpha+1).$ Computing an MST takes $O(D+\sqrt{n}\log^{*}{n})$ rounds using the algorithm of Kutten and Peleg \cite{kutten1998fast}, and computing the augmentation takes $O(T_n)$ rounds, which completes the proof.
\end{proof}

\paragraph{Roadmap:} Section \ref{sec:FirstAlgorithm-Overview} contains high-level overview of our first algorithm, full details and proofs appear in Section \ref{sec:imp}. Our Second algorithm appears in Section \ref{app:secondAlgo}.

\vspace{-5pt}
\section{Overview of our first algorithm}
\label{sec:FirstAlgorithm-Overview}

\vspace{-5pt}
As explained in Section \ref{sec:tech_overview}, our approach is to simulate a parallel algorithm for set cover to approximate weighted TAP on a related virtual graph $G'$, which translates to an approximate solution for weighted 2-ECSS in the input graph $G$.
We next provide a high-level overview of the algorithm. Full details and proofs appear in Section \ref{sec:imp}.

\subsection{Working with virtual edges}


As a first step, we replace our input graph $G$ by a virtual graph $G'$ described in \cite{DBLP:conf/opodis/Censor-HillelD17, khuller1993approximation}, such that all the non-tree edges in the graph are between ancestors and descendants. 
This is done by replacing any non-tree edge $e=\{u,v\}$ that is not between an ancestor to its descendant, by the two edges $\{u,w\},\{v,w\}$ where $w$ is the lowest common ancestor (LCA) of $u$ and $v$ in the tree. As shown in \cite{DBLP:conf/opodis/Censor-HillelD17}, finding an $\alpha$-approximation for weighted TAP in the virtual graph $G'$, gives a $2\alpha$-approximation for weighted TAP in $G$, where the extra 2 factor comes from the duplication of edges.
The main ingredient that allows us to build the virtual graph and work with the virtual edges is LCA labels. 
As shown in \cite{DBLP:conf/opodis/Censor-HillelD17}, it is possible to compute LCA labels and build the virtual graph $G'$ in $O(D+\sqrt{n}\log^*{n})$ rounds. 
To simulate a distributed algorithm in $G'$, for each virtual edge we have a vertex, which is the descendant of the virtual edge, that simulates it during the algorithm. For simplicity of presentation, in the description of the algorithm we say that edges do some computations. When we say this, we mean that the vertices that simulate the edges do the computations. For more details about the virtual edges see Section \ref{sec:virtual}.

\vspace{-5pt} 
\subsection{Decomposing the tree into layers} \label{sec:layers:overview}

Our algorithm works in a graph where all the non-tree edges are between ancestors to descendants. We next explain how we exploit this structure. Intuitively, this allows us to replace all the non-tree edges that cover a tree edge by only two edges that cover roughly the same tree edges. This property is called in \cite{agarwal2018set} the \emph{small neighbourhood cover} property. A key component in the algorithm is a certain layering of the tree, we describe next. We say that a vertex is a \emph{junction} if it has more than one child in the tree. Each layer is composed of disjoint paths in the tree, as follows. The first layer consists of all the tree paths between a leaf to its first ancestor that is a junction. To define the second layer, we first contract all the paths of the first layer, and get a new tree $T_2$. Note that several paths with the same ancestor can be contracted to the same vertex. The second layer is composed of all the tree paths between leaves in $T_2$ to their first ancestors that are junctions in $T_2$. We continue in the same manner, to define all the layers. See Figure \ref{layering_overview} for an illustration. Since a leaf in layer $i$ is a junction in layer $i-1$, which has at least two leaves in the subtree rooted at it in $T_{i-1}$, it follows that the number of layers is $O(\log{n})$.\\[-7pt]   

\setlength{\intextsep}{0pt}
\begin{figure}[h]
\centering
\setlength{\abovecaptionskip}{-2pt}
\setlength{\belowcaptionskip}{6pt}
\includegraphics[scale=0.45]{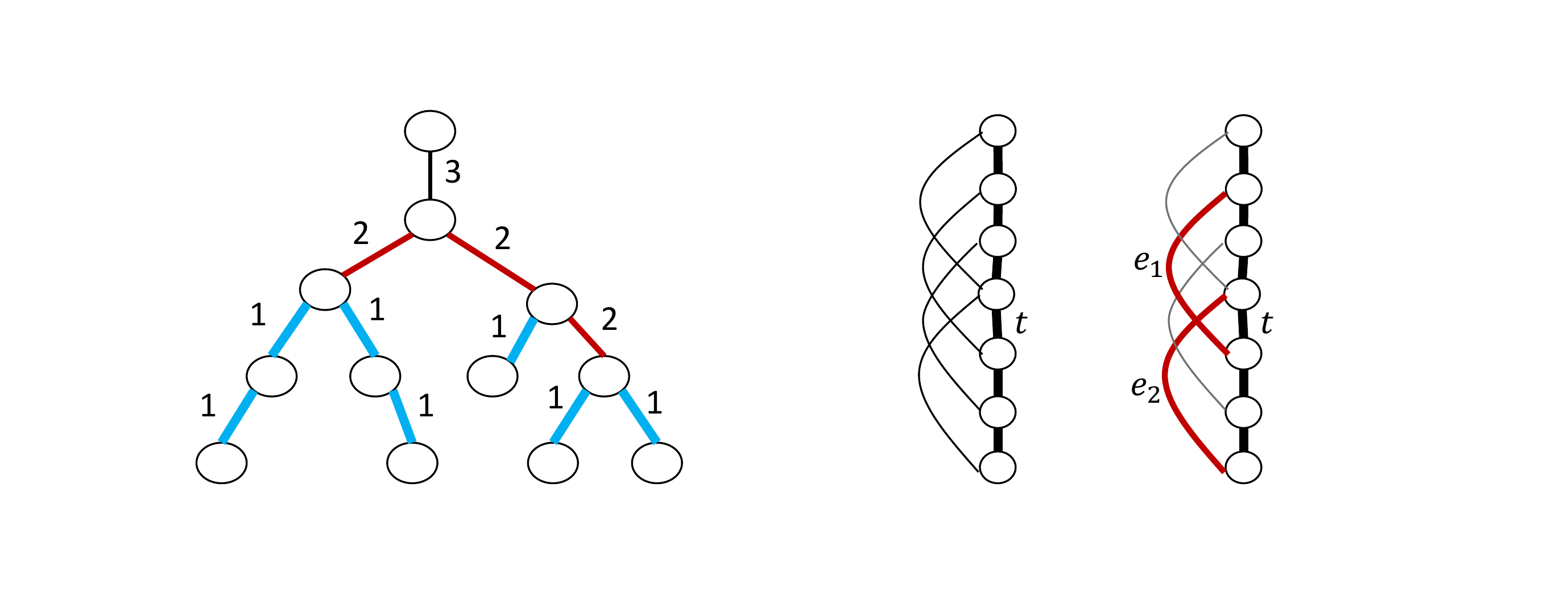}
 \caption{\small An illustration of the layering. On the left, there is a tree decomposed into layers. On the right, there is a tree (the path with bold edges) and non-tree edges that cover it, with an example of a tree edge $t$ and its two petals $e_1,e_2$.\vspace{-5pt}}
\label{layering_overview}
\end{figure}

\paragraph{The petals of a tree edge.}
We say that two tree edges $t_1,t_2$ are \emph{neighbours} with respect to a subset of non-tree edges $X$, if there is an edge $e \in X$ that covers both $t_1$ and $t_2$. The layering is useful in our algorithm for the following reason. If we look at a tree edge $t$ in layer $i$ and a subset of non-tree edges $X$, we can choose only two non-tree edges from $X$ that cover $t$ and all its neighbours in layers $j$ for $j \geq i$. These edges are called the \emph{petals} of $t$ with respect to $X$. If $t$ is an edge in the first layer, its petals are the two non-tree edges from $X$ that cover $e$ and get to the highest ancestor or the lowest descendant possible. See Figure \ref{layering_overview} for an illustration. For a tree edge $t$ in layer $i>1$, the definition is slightly more involved. The first petal is again the non-tree edge that covers $t$ and gets to the highest ancestor, but the second petal cannot be defined just as the non-tree edge that covers $t$ and gets to the lowest descendant, because $t$ may be in different leaf to root paths, and we cannot compare edges that cover different paths directly. To define the second petal, we denote by $P$ the tree path in layer $i$ that $t$ belongs to, and compare edges with respect to $P$. The second petal is defined to be a non-tree edge that covers $t$ and maximum number of edges in $P$ below $t$. For a formal definition see Section \ref{sec:layers}.

\paragraph{Computing the layers.} In Section \ref{sec:layers}, we show how to construct the layering in $O((D+\sqrt{n})\log{n})$ rounds, such that at the end all the tree edges know their layer number and some additional information about the layering. The high-level idea is to use a decomposition from \cite{DBLP:conf/podc/Dory18, ghaffari2016near} we describe in Section \ref{sec:dec}, and break the computation into several aggregate functions on trees we can compute in $O(D+\sqrt{n})$ rounds using the decomposition.   
Later, we show that in $O(D+\sqrt{n})$ rounds all the tree edges in layer $i$ can learn their petals with respect to a subset of non-tree edges $X$. To compute the higher petal we need to compute the edge that gets to the highest ancestor from the edges that cover a tree edge. This is an aggregate function of the non-tree edges that cover a tree edge, and we show how to compute such functions efficiently in Section \ref{sec:agg}. Computing the lower petal is more involved, and we show how certain information about the layers and the LCA labels allow to compute it as well. 
For full details and proofs about the layering see Section \ref{sec:layers}. 

\subsection{The parallel set cover algorithm}
\vspace{-5pt}

We follow a parallel algorithm for set cover instances with the \emph{small neighbourhood cover} property \cite{agarwal2018set}. It is easy to implement the algorithm in a setting where the complete set cover graph is known: where sets can communicate directly with the elements they cover, and vice verse. In our case, the main obstacle is to show how to simulate an algorithm in our setting, where tree edges cannot communicate directly with non-tree edges that cover them and vice verse (these edges may be far from them in the graph, and they do not even know the identity of these edges). To overcome this, the main ingredient is a certain decomposition of the tree into segments from  \cite{DBLP:conf/podc/Dory18, ghaffari2016near}. We show in Section \ref{sec:agg}, how all the non-tree edges can compute aggregate functions of the tree edges they cover, and how tree edges can compute aggregate functions of the non-tree edges that cover them efficiently using the decomposition. In addition, we show that many of the computations in the algorithm are based on such aggregate functions. Similar computations are done in \cite{DBLP:conf/podc/Dory18}, and the main challenge in our case is to show that the algorithm extends also to the case of working with virtual edges.
We mention that in \cite{agarwal2018set}, the authors describe also a distributed algorithm, but it is only for the setting that the set cover graph is known and also requires very large messages, and hence is not suitable for the $\mathsf{CONGEST}$ model.
  
\paragraph{The algorithm.} The algorithm is a primal-dual algorithm that processes the graph according to the layers. The primal LP has a variable $x(e)$ for each non-tree edge, which indicates whether $e$ is added to the augmentation $A$, and its goal is to add a minimum weight set of edges to $A$ while covering all tree edges. For a non-tree edge $e$, we denote by $S_e$ all the tree edges covered by $e$. The primal and dual LPs are formulated as follows.

\begin{center}
\setlength\fboxrule{1pt}
\begin{tabular}{cc}
  \textbf{Primal LP} & \textbf{Dual LP} \\
  \fcolorbox{red!50!black}{white}{$
  \begin{aligned}
    
\min\ \sum_{e \in E \setminus T} x(e) \cdot w(e) \\
\forall t \in T, \ \sum_{e: t \in S_e} x(e) \geq 1 \\
x(e) \geq 0
  \end{aligned}
  $} &
  \fcolorbox{red!50!black}{white}{$
  \begin{aligned}
   \max\ \sum_{t \in T} y(t) \\
\forall e \in E \setminus T, \ \sum_{t \in S_e} y(t) \leq w(e)\\
y(t) \geq 0
  \end{aligned}
  $} \\
\end{tabular}
\end{center}

The algorithm is composed of two parts, a forward phase and a reverse-delete phase. The goal of the forward phase is to cover all the tree edges while making sure that all the dual constraints hold up to $(1+\epsilon')$ factor, and for all the edges added to the augmentation $A$ it holds that $\sum_{t \in S_e} y(t) \geq w(e)$. 
Intuitively, $y(t)$ can be seen as a price an edge $t$ can pay for the non-tree edges that cover it, and we add a non-tree edge $e$ to the augmentation only if the tree edges in $S_e$ pay at least $w(e)$ in total. In some sense, this guarantees that we do not add edges that are too expensive to $A$, if there is an alternative option to cover the same tree edges. However, there is no bound on the number of edges we add to $A$ in this phase.
In the reverse-delete phase the goal is to remove some of the edges from $A$, such that all the tree edges are covered and any tree edge $t$ where $y(t) > 0$ is covered at most $c$ times for a constant $c$. This guarantees a $(c+\epsilon)$-approximation, as follows. 

\begin{lemma} \label{approx_lemma}
If each tree edge with $y(t) > 0$ is covered at most $c$ times, the algorithm guarantees a $(c+\epsilon)$-approximation.
\end{lemma}

\begin{proof}
Let $B \subseteq A$ be the final cover obtained by the end of the reverse-delete phase. 
Since all the edges added to $A$ satisfy $w(e) \leq \sum_{t \in S_e} y(t)$ and all the edges with $y(t)>0$ are covered at most $c$ times by $B$, we get
$$w(B) = \sum_{e \in B} w(e) \leq \sum_{e \in B} \sum_{t \in S_e} y(t) \leq c \sum_{t \in T} y(t) = c(1+\epsilon') \sum_{t \in T} \frac{y(t)}{(1+\epsilon')} \leq c(1+\epsilon')OPT,$$
The last inequality follows from the fact that all the dual constraints hold up to a $(1+\epsilon')$ factor, which shows that dividing the values of $y$ by $(1+\epsilon')$ give a feasible dual solution, and from the weak duality theorem, as follows. By the weak duality theorem any feasible solution to the dual problem has value smaller or equal to the value of a feasible primal solution. Since the optimal augmentation is a set of edges that cover all the tree edges it is a feasible primal solution, and the inequality follows. Choosing $\epsilon' = \frac{\epsilon}{c}$ gives a $(c+\epsilon)$-approximation. 
\end{proof}

We start by showing an algorithm where $c=4$, and then show an improved algorithm with $c=2$.
We next explain the two phases of the algorithm.

\remove{
In the reverse-delete phase the goal is to remove some of the edges from $A$, such that all the tree edges are covered and any tree edge $t$ where $y(t) > 0$ is covered at most 4 times. This guarantees a $(4+\epsilon)$-approximation, as follows. 
Let $B \subseteq A$ be the final cover obtained by the end of the reverse-delete phase. 
Since all the edges added to $A$ satisfy $w(e) \leq \sum_{t \in S_e} y(t)$ and all the edges with $y(t)>0$ are covered at most 4 times by $B$, we get
$$w(B) = \sum_{e \in B} w(e) \leq \sum_{e \in B} \sum_{t \in S_e} y(t) \leq 4 \sum_{t \in T} y(t) = 4(1+\epsilon') \sum_{t \in T} \frac{y(t)}{(1+\epsilon')} \leq 4(1+\epsilon')OPT,$$
The last inequality follows from the fact that all the dual constraints hold up to a $(1+\epsilon')$ factor, which shows that dividing the values of $y$ by $(1+\epsilon')$ give a feasible dual solution, and from the weak duality theorem. For more details see Section \ref{sec:parallel_sc}. Choosing $\epsilon' = \frac{\epsilon}{4}$ gives a $(4+\epsilon)$-approximation. We next explain the two phases of the algorithm.
}
 
\subsection{The forward phase}

The goal of the forward phase is to cover all the tree edges while making sure that all the dual constraints hold up to $(1+\epsilon)$ factor, and for all the edges added to the augmentation $A$ it holds that $\sum_{t \in S_e} y(t) \geq w(e)$. We process the layers one by one according to their order, where in epoch $k$ we make sure that all the tree edges of layer $k$ are covered. Epoch $k$ works as follows. Let $R_k$ be all the tree edges of layer $k$ that are still not covered, they are the only edges which increase their dual variables in epoch $k$. For a non-tree edge $e$, let $s(e) = \sum_{t \in S_e} y(t)$ be the current value of the dual constraint, and let $S^k_e$ be the tree edges in $R_k \cap S_e$, if each one of these tree edges $t$ sets $y(t)=\frac{w(e)-s(e)}{|S^k_e|}$, then the dual constraint becomes tight. Since we want to maintain feasibility of the dual, each tree edge $t \in R_k$ sets its dual variable to be $\min_{t \in S_e} \frac{w(e)-s(e)}{|S^k_e|}$. After this, we add to $A$ edges $e$ if their dual constraint becomes tight. In the next iteration, each tree edge in $R_k$ that is still not covered by $A$, increases its dual variable by a multiplicative factor of $(1+\epsilon)$, and again we add edges to $A$ if their dual constraint becomes tight, we continue in the same manner until all the edges in $R_k$ are covered. 
This concludes the description of the forward phase.

Our algorithm differs slightly from the algorithm in \cite{agarwal2018set}, where they maintained the feasibility of the dual solution. Our approach allows us to get an improved approximation of $(1+\epsilon)c$, compared to the $(2+\epsilon)c$-approximation obtained by the approach in \cite{agarwal2018set}.

\paragraph{Correctness.}
From the above description, it follows that at the end all the tree edges are covered by $A$, and for each $e \in A$, its dual constraint becomes tight. For all the dual constraints it holds that $\sum_{t \in S_e} y(t) \leq (1+\epsilon) w(e)$, since we increase the value of $\sum_{t \in S_e} y(t)$ at most by $(1+\epsilon)$ factor at each iteration, and $e$ is added to $A$ once $\sum_{t \in S_e} y(t) \geq w(e)$. At this point, all the edges in $S_e$ are covered and do not increase their dual variables anymore.

\paragraph{Implementation details.}
In Section \ref{sec:forward-imp}, we explain how to implement each iteration of the forward phase in $O(D+\sqrt{n})$ rounds. As already mentioned, the main building blocks we use are algorithms for computing aggregate functions of tree edges or non-tree edges described in Section \ref{sec:agg}. We show that all the computations in the forward phase are based on such functions. For example, non-tree edges should compute the value $\sum_{t \in S_e} y(t)$ which is an aggregate function of the tree edges they cover. Tree edges should learn if they are covered by $A$, which happens iff there is at least one edge that covers them that is added to $A$, this is an aggregate function of the tree edges that cover them, etc. Since there are $O(\log{n})$ layers, there are $O(\log{n})$ epochs. We show that each of them consists of $O(\frac{\log{n}}{\epsilon})$ iterations, based on the fact that uncovered tree edges increase their dual variable by a $(1+\epsilon)$ factor in each iteration. This results in a time complexity of $O((D+\sqrt{n})\frac{\log^2{n}}{\epsilon})$ rounds for the whole phase.
Full details and proofs appear in Section \ref{sec:forward-imp}. 
  
\subsection{The reverse-delete phase} \label{sec:rev-del}
In the reverse-delete phase we choose a subset $B \subseteq A$ with the following properties. 

\begin{enumerate}
\item $B$ covers all the tree edges.
\item Any $t \in T$ with $y(t) > 0$ is covered at most 4 times by $B$.
\end{enumerate}
 
The dual variables are not changed during the process. 
Recall that $R_k$ are all the tree edges in layer $k$ that are not covered before epoch $k$. From the description of the forward phase, the only tree edges with $y(t)>0$ are in $R_k$ for some $k$. We need the following additional definitions.
We denote by $A_k$ the non-tree edges added to $A$ in epoch $k$, and by $F_k$ the tree edges that are first covered in epoch $k$. 
Note that $F_k$ contains $R_k$ and perhaps additional edges from higher layers.
By the definition of $R_k$, edges in $R_k$ are not covered by $A_i$ for any $i < k$. This motivates going over the layers in the reverse direction. In the reverse-delete phase we go over the layers in the reverse direction, building $B$. Initially $B = \emptyset$ and we add to it edges from $A$ during the algorithm. Our algorithm proceeds in epochs $k=L,...,1$, where $L$ is the index of the last layer. We next describe epoch $k$.

\paragraph{Epoch $k$.}
In epoch $k$ we make sure that the following holds. 

\begin{enumerate}
\item All the tree edges in $F_i$ for $i \geq k$ are covered by $B$. 
\item All the edges in $R_i$ for $i \geq k$ are covered at most 4 times. 
\end{enumerate}

Hence, at the end of epoch $1$, $B$ satisfies all the requirements. 
We next explain how we guarantee the above properties. At the beginning of epoch $k$, $B$ already covers all the edges in $F_i$ for $i>k$, and our goal is to add to it edges from $A_k$ to cover $F_k$. We need to make sure that all the edges in $R_i$ for $i \geq k$ are covered at most 4 times. Note that $B$ already satisfies this for $i>k$, and by definition edges in $A_k$ do not cover edges in $R_i$ for $i>k$, so we need to take care only of edges in $R_k$. However, edges in $R_k$ may already be covered by $B$, so we may need to remove edges from $B$, while making sure that all the edges in $F_i$ for $i \geq k$ are covered. For doing so, we go over the layers $i=k,...,L$ where $L$ is the last layer. We need the following notation.

Let $X = B \cup A_k$, let $F = \cup_{i=k}^{L} F_i$, and let $H_i$ be all the tree edges in $F$ in layer $i$ (note that $F_i$ may contain also edges from layers higher than $i$). Notice that $B$ covers $F_i$ for $i>k$, and $A_k$ covers $F_k$, hence $X$ covers $F$. Let $t$ be a tree edge in $H_i$, and let $p(t)$ be the set of non-tree edges in $X$ that cover $t$. As explained in Section \ref{sec:layers:overview}, we can replace $p(t)$ by two non-tree edges in $p(t)$ that cover all the tree edges covered by $p(t)$ in layers $i,...,L$, which are called the \emph{petals} of $t$ in $X$.
We next build a set $Y \subseteq X$ that covers $F$. We go over the layers in iterations $i=k,...,L$, as follows.

\paragraph{Iteration $i$.}
In iteration $i$ we make sure that the edges in $H_i$ are covered. Let $\widetilde{H}_i$ be all the tree edges in $H_i$ that are not covered by $Y$ at the beginning of iteration $i$. We build a virtual graph $G_i$ that its vertices are the tree edges in $\widetilde{H}_i$, and two vertices $t_1,t_2$ in $G_i$ are connected if there is a non-tree edge in $X$ that covers $t_1,t_2$. We find a maximal independent set (MIS), $M_i$, in the graph $G_i$. We call the elements in $M_i$ anchors and add all their petals (with respect to $X$) to $Y$. Intuitively, the computation of the MIS, $M_i$, allows us to cover simultaneously all the edges in $\widetilde{H}_i$, while making sure that we do not cover edges too many times.\\[-7pt] 

After going over all the layers $i=k,...,L$, we set $B = Y$. This completes the description of epoch $k$. 

\paragraph{Correctness proof.}
The correctness proof follows \cite{agarwal2018set}, we include it for completeness.

\begin{lemma} \label{correct_lemma}
At the end of epoch $k$:
\begin{enumerate}
\item All the tree edges in $F = \cup_{i=k}^{L} F_i$ are covered by $B$.\label{cover_F} 
\item All the edges in $R_i$ for $i \geq k$ are covered at most 4 times.\label{cover_R} 
\end{enumerate}
\end{lemma}

\begin{proof}
The proof is by induction on $k$. In the base case, $k=L+1$ and $B=\emptyset$, hence the claim clearly holds.
We start by proving \ref{cover_F}.
From the description of iteration $i$, at the end of iteration $i$, $Y$ covers $H_i$: All the edges in $H_i \setminus \widetilde{H}_i$ are already covered at the beginning of the iteration. All the edges in $\widetilde{H}_i$ are either anchors or have a neighbouring anchor $t \in M_i$, and then they are covered by the petals of $t$. Hence, after all the iterations, $B=Y$ covers $\cup_{i=k}^{L} H_i = F$.

We next prove \ref{cover_R}. We start with a simple observation: all the anchors in all the layers, $\cup_{i=k}^{L} M_i$, are independent in the following sense. If we take any two anchors $t_1,t_2$, there is no edge in $X$ that covers both of them. In the same layer, it follows from computing an MIS. In different layers $i<j$, it follows since the graph $G_j$ contains only tree edges that are not yet covered, and we add the petals of all the anchors in $M_i$ at the end of iteration $i$.

We next show that any tree edge $t \in R_i$ for $i \geq k$ is covered at most 4 times by $Y$. As explained earlier, for $i>k$ this already follows from previous epochs. This holds since $Y \subseteq X$, and $X=B \cup A_k$ (with the set $B$ at the end of the previous epoch). Now, at the end of epoch $k+1$, all the edges in $R_i$ for $i>k$ are covered at most 4 times by $B$ from the induction hypothesis. Also, from the definition of the sets, $A_k$ does not cover any edge in $R_i$ for $i>k$. Hence, any set $Y \subseteq X$ we choose covers all the edges in $R_i$ for $i>k$ at most 4 times. 

We now show that any tree edge $t \in R_k$ is covered at most 4 times by $Y$. If $t$ is an anchor, it is covered only by its petals, since the anchors are independent, and hence it is covered at most twice. Otherwise, we show that $t$ has at most two neighbouring anchors. Since $t$ is in layer $k$, and all the anchors are at layers at least $k$, the 2 petals of $t$ cover all its anchors. Hence, if there are more than two anchors, at least two of them are covered by the same petal, which contradicts the fact that the anchors are independent. Hence, $t$ is covered at most twice by the petals of each of its neighbouring anchors, and at most 4 times in total. This completes the proof.
\end{proof}

From Lemma \ref{correct_lemma} with respect to $k=1$, we get that at the end of the reverse-delete phase $B$ covers all tree edges, and for each $1 \leq i \leq L$, all the tree edges in $R_i$ are covered at most 4 times, as needed. The algorithm has $L$ epochs, each of them consists of at most $L$ iterations, which sums up to $O(L^2)=O(\log^2{n})$ iterations. In Section \ref{sec:rev-del-imp}, we show how to implement each iteration in $O(D+\sqrt{n})$ rounds, which results in a complexity of $O((D+\sqrt{n})\log^2{n})$ rounds for the whole phase.\\[-7pt]

\paragraph{Implementation details.} To implement the algorithm, the main task we need to solve is to build an MIS in the virtual graph $G_i$ in each iteration. This is a \emph{completely virtual} graph, and we design a specific algorithm to solve this task. The MIS, $M_i$, is a set of tree edges in $\widetilde{H}_i$ such that for any two tree edges $t,t' \in M_i$, there is no edge in $X$ that covers both of them. Our algorithm is composed of a global part where all the vertices locally compute a global MIS that includes some of the tree edges. Then, it has a local part, where we work on each segment separately, and add additional edges to the MIS by scanning the segment. 
 To implement the global part, first, all the vertices in the graph learn $O(\log{n})$ information about each segment: two of its tree edges and their petals. Let $T'$ be all the tree edges learned. We show that based on this information, all the vertices can compute locally an MIS, $M'$, of the tree edges $T'$. In the local part, we work in different segments in parallel and the main challenge is to show that all the edges added to the MIS are indeed independent. Full details and proofs appear in Section \ref{sec:rev-del-imp}.
  
\paragraph{Conclusion.} Based on these ingredients we get a $(4+\epsilon)$-approximation for weighted TAP in the virtual graph $G'$, which translates to a $(8+\epsilon)$-approximation for weighted TAP in the input graph $G$, and a $(9+\epsilon)$-approximation for weighted 2-ECSS in $G$. The complexity of the whole algorithm is $O((D+\sqrt{n})\frac{\log^2{n}}{\epsilon})$ rounds. We next give the high-level idea for improving the approximation. 

\subsection{Improved approximation}

In the reverse-delete phase we make sure that all the edges in $R_k$ are covered at most $4$ times. We design a variant of the algorithm where all these tree edges are covered at most 2 times, which results in an improved approximation for the whole algorithm. The first change we do in the algorithm is that we replace the MIS, $M_i$, by a set of tree edges $M'_i$ that are no longer guaranteed to be independent, however for each one of the tree edges in $M'_i$ we add only its higher petal to $Y$ and not both of them. Then, we are able to analyze the structure of dependencies in $M'_i$ and show that all the tree edges in $R_k$ are covered at most 3 times by the edges of $Y$. To improve the approximation further we show that the only cases where tree edges are covered 3 times have a certain structure, and in this case we can remove one of the edges that cover them from $Y$ without affecting the covering of all the other tree edges. We also show that we can detect these cases efficiently. This results in an algorithm where all the tree edges in $R_k$ are covered at most 2 times. Based on this we get a $(2+\epsilon)$-approximation for weighted TAP in the virtual graph $G'$, which translates to a $(4+\epsilon)$-approximation for weighted TAP in the input graph $G$, and a $(5+\epsilon)$-approximation for weighted 2-ECSS in $G$. The complexity remains $O((D+\sqrt{n})\frac{\log^2{n}}{\epsilon})$ rounds. Full details and proofs appear in Section \ref{sec:improved}.

\subsubsection{A note on the unweighted case} \label{sec:unweighted_approx}

While obtaining a $(2+\epsilon)$-approximation for weighted TAP on the virtual graph $G'$ requires an involved analysis, obtaining a $2$-approximation for the unweighted variant of the same problem is easy: we start by computing an MIS of the tree edges (with respect to all the non-tree edges). Then, for each one of the tree edges we add its 2 petals to the augmentation. Computing an MIS of the tree edges can be done as explained in Sections \ref{sec:rev-del} and \ref{sec:rev-del-imp} where we proceed according to the layers. At the end, all the tree edges are covered: consider a tree edge $t$ in layer $i$, if it is already covered at the beginning of iteration $i$ it is clear. Otherwise, it has a neighbor $t'$ in the MIS in layer $i$, and the petals of $t'$ cover $t$. The algorithm guarantees a $2$-approximation for unweighted TAP since we must add at least one edge to the augmentation to cover each one of the edges in the MIS (since they are independent, there is no edge that covers two such edges), and the algorithm adds exactly 2 edges to the augmentation for each edge of the MIS, which guarantees a $2$-approximation. 

This allows to give a $4$-approximation for unweighted TAP in the input graph $G$ in $\widetilde{O}(D+\sqrt{n})$ rounds. While the same result already appears in \cite{DBLP:conf/opodis/Censor-HillelD17}, the approximation analysis in \cite{DBLP:conf/opodis/Censor-HillelD17} is quite involved, where the analysis here is very simple. Also, it seems that the same approach can be used to show a $\tau$-approximation for \emph{unweighted} problems with the $\tau$-SNC property, which is better than the approximation shown in \cite{agarwal2018set} (this algorithm already computes an MIS of the elements we need to cover, but the approximation ratio analysis is for the weighted case which gives a worse approximation). This can be seen as an extension of the simple approximation algorithm for minimum vertex cover that starts by computing a maximal matching and then adds all the endpoints of the edges of the matching to the cover. In this setting, the elements we need to cover are the edges of the graph, an MIS of them is a matching, and the petals of an edge are its endpoints.

\section{The first algorithm: full details and proofs} \label{sec:imp}

\subsection{Working with virtual edges} \label{sec:virtual}

In our algorithm, we replace the input graph $G$ with a virtual graph $G'$. The graph $G'$ includes all the tree edges in $G$ and each non-tree edge in $G$ is replaced by one or two \emph{virtual} edges in $G'$ such that all the non-tree edges in $G'$ are between ancestors to descendants. We next explain how we construct the graph $G'$ and work with it in the algorithm. We follow the distributed construction in \cite{DBLP:conf/opodis/Censor-HillelD17}. The graph $G'$ was first defined in the sequential algorithm of Khuller and Thurimella \cite{khuller1993approximation}.

The graph $G'$ is defined as follows. Let $e=\{u,v\}$ be a non-tree edge in $G$. If $e$ is an edge between an ancestor and a descendant, we keep it on $G'$. Otherwise, let $w=LCA(u,v)$ be the lowest common ancestor of $u$ and $v$ in the tree. The edge $e$ is replaced by the two virtual edges $\{u,w\},\{v,w\}$. These edges are clearly between an ancestor, $w$, to its descendants, $u$ and $v$. In our algorithm, the descendants $u$ and $v$ would know about the virtual edges $\{u,w\},\{v,w\}$ and would simulate them in the algorithm. However, the ancestor $w$ does not necessarily know about these edges. 

The edge $\{u,v\} \in G$ covers the tree path $P_{u,v}$ between $u$ and $v$, that is composed of the two paths $P_{u,w},P_{v,w}$ where $w=LCA(u,v)$. Hence, it follows that an edge $e \in G$ is replaced by one or two edges in $G'$ that cover exactly the same tree edges. Based on this, it is proved in \cite{DBLP:conf/opodis/Censor-HillelD17} that an $\alpha$-approximation augmentation $A'$ in $G'$ gives a $2\alpha$-approximation augmentation $A$ in $G$. The augmentation $A$ is constructed from $A'$ by replacing any edge in $A'$ by a corresponding edge in $A$ as follows. If $e' \in A'$, the edge $e'$ was added to $G'$ to replace an original edge $e \in G$, the edge $e$ is the corresponding edge. If there are several such edges $e$, we choose one of them.

\begin{lemma} \label{virtual_lemma}
Let $A'$ be an $\alpha$-approximation augmentation in $G'$, and let $A$ be the corresponding augmentation in $G$ where any edge in $A'$ is replaced by a corresponding edge in $A$. Then, $A$ is a $2\alpha$-approximation augmentation in $G$.
\end{lemma}  

To construct $G'$, we need to compute all the virtual edges in $G'$. For this, we follow the distributed construction in \cite{DBLP:conf/opodis/Censor-HillelD17} (see Section 5.2 in the full version of the paper), 
that uses the LCA labelling scheme of Alstrup et al. \cite{alstrup2004nearest}. In \cite{DBLP:conf/opodis/Censor-HillelD17}, it is shown how to assign the vertices of the graph short labels of $O(\log{n})$ bits, such that given the labels and some additional information about the structure of the tree, any two vertices $u,v$ can compute the label of their LCA from their labels. The time complexity for assigning the labels and learning the structure of the tree is $O(D+\sqrt{n}\log^{*}n)$ rounds. Given the labels, the computation of the label of the LCA is immediate. Hence, if $\{u,v\}$ is an edge in $G$, both the vertices $u$ and $v$ know their own labels and would learn the label of $w=LCA(u,v)$, and use it to simulate the virtual edges $\{u,w\},\{v,w\}.$ This is summarized in the next Lemma.

\begin{lemma}
Building the virtual graph $G'$ takes $O(D+\sqrt{n}\log^{*}{n})$ rounds. At the end, for each virtual edge $e=\{u,w\}$ where $w$ is an ancestor of $u$, the vertex $u$ knows that $e$ is in $G'$ and knows the LCA labels of $u,w$.
\end{lemma}

During our algorithm we use the LCA labels to replace the original ids of vertices. This is useful for computing some tasks. For example, given the labels of a tree edge $t=\{v,p(v)\}$ where $p(v)$ is the parent of $v$, and given the labels of a virtual edge $e=\{anc,dec\}$ where $anc$ is an ancestor of $dec$, it is easy to check if $e$ covers $t$. We simply check if $anc$ is an ancestor of $p(v)$ and $v$ is an ancestor of $dec$, which can be deduced from the LCA labels: a vertex $v$ is an ancestor of $u$ iff $LCA(u,v)=v$. This is summarized in the following observation.

\begin{observation} \label{obs_labels}
Given the labels of a tree edge $t$ and a non-tree edge $e$ between an ancestor to its descendant, we can determine whether $e$ covers $t$. 
\end{observation}

\subsection{Computing aggregate functions} \label{sec:agg}

To implement our algorithm, we need the following building blocks. 
First, all the non-tree edges simultaneously should learn an aggregate function of the tree edges they cover.
Second, all the tree edges simultaneously should learn an aggregate function of the non-tree edges that cover them. 
We next show that these two building blocks can be implemented in $O(D+\sqrt{n})$ rounds.
The aggregate functions we consider are commutative functions with inputs and outputs of $O(\log{n})$ bits, for example: minimum, sum, etc. The important property we exploit is that we can apply these functions on different parts of the inputs separately and then combine the results. For example, when computing a sum of inputs, we can divide the inputs arbitrarily to several sets, compute the sum in each set and then sum the results.

We next explain how we compute the aggregate functions.
The computation is similar to \cite{DBLP:conf/podc/Dory18}. The main difference is that we work with virtual edges and not with the original edges of the graph, and we show that it is possible to extend the algorithm from \cite{DBLP:conf/podc/Dory18} to deal with virtual edges. The computation uses a decomposition of the tree into segments described in \cite{DBLP:conf/podc/Dory18}, which is a variant of a decomposition presented for solving the FT-MST problem \cite{ghaffari2016near}.
We start by describing the decomposition, and then explain how we use it to compute aggregate functions.

\subsubsection{Overview of the decomposition} \label{sec:dec}

Here we give a high-level overview of the decomposition, for full details see \cite{DBLP:conf/podc/Dory18, ghaffari2016near}.
The decomposition breaks the tree into $O(\sqrt{n})$ \emph{edge-disjoint} segments, each with diameter $O(\sqrt{n})$. 
Each segment $S$, has a root $r_S$ which is an ancestor of all the vertices in the segment. In addition, there is a special vertex $d_S$ which is called the \emph{unique descendant} of the segment. The structure of the segment is as follows. It has a path between $r_S$ and $d_S$, which is called the \emph{highway} of the segment, and it has additional subtrees attached to the highway. A crucial ingredient of the decomposition is that $r_S$ and $d_S$ are the \emph{only} vertices in the segment $S$ that may be contained in other segments, and other vertices in $S$ are not connected by a tree edge to any other segment. The id of a segment is the pair $(r_S,d_S).$ 
The \emph{skeleton tree} $T_S$ is a virtual tree with $O(\sqrt{n})$ vertices that captures the structure of the decomposition, as follows. For each vertex that is either $r_S$ or $d_S$ in one of the segments, there is a vertex in $T_S$, and the edges in $T_S$ correspond to the highways of the segments. I.e., there is an edge $\{u,v\} \in T_S$ where $u$ is a parent of $v$ iff $u=r_S$ and $v=d_S$ for some segment $S$.

In \cite{DBLP:conf/podc/Dory18}, it is proved that the segments can be constructed \emph{deterministically} in $O(D+\sqrt{n}\log^*{n})$ rounds.
In addition, the two following claims are proved in \cite{DBLP:conf/podc/Dory18}, where $P_{u,v}$ is the unique tree path between $u$ and $v$.

\begin{claim} \label{info1}
In $O(D+\sqrt{n})$ rounds, the vertices learn the following information. All the vertices learn the id of their segment, and the complete structure of the skeleton tree. In addition, each vertex $v$ in the segment $S$ learns all the edges of the paths $P_{v,r_S}$ and $P_{v,d_S}$.
\end{claim}

\begin{claim} \label{info2}
Assume that each tree edge $t$ and each segment $S$, have some information of $O(\log{n})$ bits, denote them by $m_t$ and $m_S$, respectively. In $O(D+\sqrt{n})$ rounds, the vertices learn the following information. 
Each vertex $v$ in the segment $S$ learns the values $(t,m_t)$ for all the tree edges in the highway of $S$, and in the paths $P_{v,r_S},P_{v,d_S}$. In addition, all the vertices learn all the values $(S,m_S)$. 
\end{claim}

\subsubsection{Computing aggregate functions of tree edges}

In this section, we explain how all the non-tree edges simultaneously learn an aggregate function of tree edges they cover: an edge $e$ learns an aggregate function of the edges $S_e$. We prove the following.

\begin{claim} \label{agg_tree}
Assume that each tree edge $t$ has some information $m_t$ of $O(\log{n})$ bits, and let $f$ be a commutative function with output of $O(\log{n})$ bits. In $O(D+\sqrt{n})$ rounds, each non-tree edge $e$, learns the output of $f$ on the inputs $\{m_t\}_{t \in S_e}.$ 
\end{claim}

In our algorithm, some of the non-tree edges are virtual. As explained in  Section \ref{sec:virtual}, the descendant of the corresponding virtual edge simulates the edge. When we say that a non-tree edge learns some information, the vertex that simulates the virtual edge learns the information. We next prove Claim \ref{agg_tree}.

\begin{proof}
Let $h_S$ be all the edges in the highway of the segment $S$, and let $m_S$ be the $O(\log{n})$-bit output of $f$ on the inputs $\{m_t\}_{t \in h_S}$. The value $m_S$ can be computed locally in the segment $S$ in $O(\sqrt{n})$ rounds by scanning $h_S$. Using Claims \ref{info1} and \ref{info2}, in $O(D+\sqrt{n})$ rounds, all the vertices learn the id of their segment, the complete structure of the skeleton tree and all the values $(S,m_S)$. In addition, each vertex $v$ in the segment $S$ learns the values $(t,m_t)$ for all the tree edges in the highway of $S$, and in the paths $P_{v,r_S},P_{v,d_S}$. 

Let $e=\{u,v\}$ be a non-tree edge in the original graph $G$. In the virtual graph $G'$, either $e$ is in the graph, which means that it is an edge between a descendant to an ancestor, or the edge $e$ is replaced with the two edges $\{u,w\},\{v,w\}$ where $w=LCA(u,v).$ We next show that $u$ learns all the information needed to evaluate $f$ on the edges in the tree path $P_{u,w}$, which is exactly the path of tree edges covered by the virtual edge $\{u,w\}$, and $v$ has all the information to evaluate $f$ on the edges in the tree path $P_{v,w}$. Exactly in the same way, if $e \in G'$, and assume w.l.o.g that $u$ is a descendant of $v$, then $u$ learns all the information needed to evaluate $f$ on the edges in the tree path $P_{u,v}$ (now $v = LCA(u,v)$). In all the cases, the vertex that simulates the non-tree edge learns the output of $f$, as needed.
The proof is divided to 3 cases. The vertices $u$ and $v$ can distinguish between the cases from knowing the id of their segment, their ancestors in the segment and their LCA.\\[-7pt]

\textbf{Case 1:} $u$ and $v$ are in the same segment $S$. In this case, $w=LCA(u,v)$ is also in the segment $S$, and both $u$ and $v$ know the label of $w$ based on their LCA labels. Additionally, the tree path $P_{u,w}$ is contained in the path $P_{u,r_S}$ where $r_S$ is the ancestor of the segment $S$. Vertex $u$ knows all the values $(t,m_t)$ for the edges $t \in P_{u,r_S}$, which allows $u$ computing the output of $f$ on the inputs $\{m_t\}_{t \in P_{u,w}}$. Note that $P_{u,w}$ is exactly the path of tree edges covered by the edge $\{u,w\}$. The same holds for $v$ with respect to $P_{v,w}$, which completes the proof.\\[-7pt]

\textbf{Case 2:} $u$ and $v$ are in different segments, and $LCA(u,v)$ is in another segment. Let $r_u$ and $r_v$ be the ancestors in the segments of $u$ and $v$, respectively. Then, $w=LCA(u,v)=LCA(r_u,r_v).$ We next show that $w$ is necessarily a vertex in the skeleton tree $T_F$. Let $S$ be the segment of $w$. Since $w = LCA(r_u,r_v)$, the first edges entering the segment $S$ in the paths $P_{r_u,w},P_{r_v,w}$ must be different. However, $r_S$ and $d_S$ are the only vertices in $S$ connected to other segments, which implies that $w$ must be $r_S$ or $d_S$. Now, $u$ and $v$ know the complete structure of the skeleton tree, and in particular can deduce the paths between $r_u$ to $w$ and between $r_v$ to $w$ in $T_F$. 
The paths $P_{r_u,w}$ and $P_{r_v,w}$ in the original tree $T$ correspond to the paths in $T_F$, where each edge in $T_F$ is replaced by the corresponding highway in $T$.
Also, for each highway $h_S$, all the vertices know $m_S$, the output of $f$ on tree edges of $h_S$. Hence, $u$ computes the output of $f$ on tree edges in $P_{u,w}$, as follows. The path $P_{u,w}$ consists of $P_{u,r_u}$ and $P_{r_u,w}$. The vertex $u$ knows all the values $(t,m_t)$ for tree edges in $P_{u,r_u}$, and all the values $m_S$ for highways in the path $P_{r_u,w}$, which allows $u$ to compute the output of $f$ on the tree edges in $P_{u,w}$. Similarly, $v$ computes the output of $f$ on the tree edges in $P_{v,w}$.\\[-7pt]

\textbf{Case 3:} $u$ and $v$ are in different segments, and $w=LCA(u,v)$ is in the same segment of one of them. Assume w.l.o.g that $w$ is in the segment of $u$. Now, since $P_{u,w} \subseteq P_{u,r_u}$ and $u$ knows all the values $(t,m_t)$ for tree edges in $P_{u,r_u}$, and it knows the label of $w$ from the LCA labels of $u$ and $v$, then $u$ can compute the output of $f$ on the tree path $P_{u,w}$. In addition, let $d_u$ be the unique descendant in the segment $S$ of $u$. Since $r_u$ and $d_u$ are the only vertices in $S$ connected to other segments, then either $w=LCA(u,v)=r_u$, or $w$ is an ancestor of $d_u$, and hence must be on the highway $h_S$. In the first case, the path $P_{v,w=r_u}$ consists of $P_{v,r_v}$ and $P_{r_v,r_u}$. Vertex $v$ knows full information about the path $P_{v,r_v}$ and the path $P_{r_v,r_u}$ is composed of entire highways, which allows $v$ to compute the value of $f$ on the tree path $P_{v,w}$ as in Case 2. In the second case, the path $P_{v,w}$ consists of the 3 paths $P_{v,r_v},P_{r_v,d_u},P_{d_u,w}$. Again, $v$ has full information about the path $P_{v,r_v}$, and the path $P_{r_v,d_u}$ consists of highways that are known to $v$ from the structure of the skeleton tree. The path $P_{d_u,w}$ is part of the highway $h_S$ in the segment of $u$, and $u$ has all the values $(t,m_t)$ for edges in $h_S$. Since $u$ also knows $w$, it can compute the value of $f$ on the subpath $P_{d_u,w}$ and send it to $v$. Now $v$ has all the information to compute $f$ on the whole path $P_{v,w}$.\\[-7pt]

The time complexity is $O(D+\sqrt{n})$ rounds for using Claims \ref{info1} and \ref{info2}. From now on all the computations are completely local or require sending one message from $u$ to $v$.
\end{proof}

\subsubsection{Computing aggregate functions of non-tree edges}

We next explain how all the tree edges simultaneously compute an aggregate function of all the non-tree edges that cover them, proving the following.

\begin{claim} \label{agg_non_tree}
Assume that each non-tree edge $e$ has some information $m_e$ of $O(\log{n})$ bits, and let $f$ be a commutative function with output of $O(\log{n})$ bits. In $O(D+\sqrt{n})$ rounds, each tree edge $t$, learns the output of $f$ on the inputs $\{m_e\}_{t \in S_e}.$ 
\end{claim}

Again, the vertex that simulates a non-tree edge $e$, has the information $m_e$.

\begin{proof}
For the proof, we follow the terminology in \cite{ghaffari2016near, DBLP:conf/podc/Dory18}, and classify the non-tree edges that cover a tree edge $t$ to 3 types: \emph{short-range}, \emph{mid-range} and \emph{long-range}, indicating if the non-tree edge $e$ has 2,1 or 0 endpoints in the segment of $t$, respectively. To learn the output of $f$, we compute $f$ separately on the short-range, mid-range and long-range edges that cover $t$, and then apply $f$ again on the results to obtain the final output. We next explain how we compute $f$ on each one of the parts. 
Note that all the non-tree edges in the graph $G'$ are between an ancestor and a descendant, and the descendant is the vertex that simulates the non-tree edge (even for a virtual edge).\\[-7pt]  

\textbf{Computing $f$ on short-range edges.} First, we let each vertex $u$ in a segment $S$ to learn all the tree edges in the tree path $P_{u,r_S}$ in $O(D+\sqrt{n})$ rounds, using Claim \ref{info1}. Note that our algorithm works with LCA labels instead of ids, so in particular $u$ knows all the LCA labels in the path $P_{u,r_S}$.
If $e$ is a short-range edge for the tree edge $t$, then the two endpoints of $e$ are in the segment of $t$. In particular, one of these endpoints, say $u$, is a descendant of $t$ and it is the vertex that simulates the edge $e$. Hence, $u$ has the information $m_e$ and it can send it to $t$. If $u$ simulates several different short-range edges that cover $t$, first it applies $f$ on all the relevant values $m_e$, and then sends only one $O(\log{n})$-bit message to $t$. 
Note that $u$ knows the LCA labels of all the non-tree edges it simulates and it knows all the LCA labels of edges in the tree path $P_{u,r_S}$, which allow to check which of the edges it simulates are short-range edges and which of them cover $t$.
To compute the output of $f$ on all the short-range edges in the segment we use convergecast. Each time a vertex receives from its children an $O(\log{n})$-bit message for $t$, first it applies $f$ on the messages received and its message for $t$ if exists, and then sends the result to its parent. 
In general, each vertex $u$ may send a message to each one of the tree edges in the path $P_{u,r_u}$ where $r_u$ is the ancestor in the segment of $u$. In order that all the vertices in the segment would send such messages to all the tree edges above them in the segment we pipeline the computations. The complexity is $O(\sqrt{n})$ rounds, since this is the diameter of the segments.\\[-7pt]

\textbf{Computing $f$ on long-range edges.} 
First, all the vertices learn the complete structure of the skeleton tree in $O(D+\sqrt{n})$ rounds, using Claim \ref{info1}. In particular, for each segment $S$ all the vertices learn the LCA labels of $r_S$ and $d_S$.
If a tree edge $t$ has a long-range edge that covers it, then $t$ must be on a highway of a segment, since otherwise all its descendants are in the segment. Also, if $e$ is a long-range edge for the tree edge $t$ in the highway $h_S$, then $e$ must be a long-range edge for all the tree edges in $h_S$, since one of the endpoints of $e$ is an ancestor of $r_S$ and the second a descendant of $d_S$ (all the non-tree edges in our graph are between an ancestor and a descendant). Since each vertex $u$ that simulates a non-tree edge $e$ has the LCA labels of $e$ and knows the complete structure of the skeleton tree, it knows exactly which highways are covered by $e$, as follows.
If $u$ is a descendant of the virtual edge $e=\{u,w\}$ and it wants to learn if this edge covers the highway $h_S$, it should check if $w$ is an ancestor of $r_S$ and $d_S$ is an ancestor of $u$, which is immediate using the LCA labels, as explained in Section \ref{sec:virtual}. 
Hence, each vertex $u$ can compute for each one of the $O(\sqrt{n})$ highways $h_S$ the value of $f$ computed on all the inputs $m_e$ where $e$ is an edge that $u$ simulates and is a long-range edge that covers the highway $h_S$. To learn the value of $f$ computed on all the inputs $m_e$ where $e$ is a long-range edge that covers the highway $h_S$ we use convergecast over the BFS tree. To learn these values for all the $O(\sqrt{n})$ highways we use pipelining. The overall complexity is $O(D+\sqrt{n})$ rounds, and at the end all the vertices know the output of $f$ computed on long-range edges that cover $h_S$ for each one of the highways.\\[-7pt]  

\textbf{Computing $f$ on mid-range edges.} Let $t$ be a tree edge in a segment $S$ covered by a mid-range edge $e$. One of the endpoints of $e=\{u,v\}$, say $u$, is a descendant of $t$. One case it that $u$ is in the segment of $t$, and can send the value $m_e$ to $t$. In general, all the tree edges can learn the output of $f$ on mid-range edges that cover them with a descendant in their segment exactly as they learn the output of $f$ on short-range edges. 

We next focus on the case that $u$ is in a different segment, in this case $u$ must be a descendant of $d_S$. Since $e$ is a mid-range edge, it follows that $v$ is in the segment of $t$. However, $e$ could be a virtual edge, simulated by the descendant $u$, and $v$ does not necessarily know about the edge $e$. If $v=r_S$, then the edge $e$ covers the whole highway of $S$, and we can treat it as a long-range edge ($u$ would include the input $m_e$ when it computes $f$ on long-range edges that cover $h_S$). Otherwise, let $e'=\{u,w\}$ be the original edge that was replaced by the virtual edge $e$. From the structure of virtual edges and the decomposition, it holds that $LCA(u,w)=v$, and $w$ is in the segment $S$. This follows since $v$ is a vertex in $S$ which is not $r_S$ or $d_S$, and hence does not have any descendant outside $S$ that is not a descendant of $d_S$. Also, $v$ is on the highway $h_S$ since it is an ancestor of $d_S$. Now, $u$ can send to $w \in S$ the value $m_e$ (this requires exchanging one message over the original edge $\{u,w\}$), and we would like $w$ to let $t$ learn the value $m_e$. 

This is done as follows. For a vertex $x \neq \{d_S,r_S\}$ on the highway $h_S$, $x$ has exactly one child $y$ on the highway, we denote by $T_x$ the subtree of $x$ in $S$ excluding the edge $\{x,y\}$ and the subtree rooted at $y$. All the vertices $w' \in T_x$, have $LCA(w',d_S)=x$, and hence any original edge of the form $\{w',u'\}$ where $w' \in T_x$ and $u'$ is a descendant of $d_S$ in another segment, translates to the virtual edge $\{x,u'\}$, and in particular is a mid-range edge for all the edges in the path $P_{x,d_S}$. We would like first to let $x$ learn the output of $f$ computed on the inputs $m_e$ where $e=\{u',w'\}$ is of the form described above. This is done by scanning the tree $T_x$, where each vertex $w' \in T_x$ computes $f$ on the mid-range edges it knows and on the outputs it receives from its children in the tree. Note that the trees $T_x$ are disjoint for different vertices $x$, hence the computations are done in parallel for different trees. Now, we scan the highway $h_S$, from $r_S$ to $d_S$. Edges that cover the whole highway were already treated as global edges, so we start with the child of $r_S$, call it $x$. Vertex $x$ knows the value of $f$ on all mid-range edges that cover $P_{x,d_S}$ and passes this information to its child $x'$. The value that $x$ computed is exactly the value of $f$ relevant for the tree edge $t=\{x',x\}$. For the tree path $P_{x',d_S}$, any mid-range edge that covers it either covers also the path $P_{x,d_S}$ (all the edges that cover $P_{x,d_S}$ cover also the subpath $P_{x',d_S}$) or covers only the subpath $P_{x',d_S}$, in this case it has an original endpoint $w' \in T_{x'}$.
Since $x'$ knows the value of $f$ computed on all mid-range edges that cover the tree path $P_{x',d_S}$ that have one original endpoint in $T_{x'}$, and $x'$ received from $x$ the value of $f$ computed on all mid-range edges that cover the tree path $P_{x,d_S}$, computing $f$ on these two outputs results in the output of $f$ on all mid-range edges that cover the tree path $P_{x',d_S}$. We continue in the same way, and eventually all the tree edges in $h_S$ learn the relevant output of $f$. This is done in all the segments in parallel which takes $O(\sqrt{n})$ rounds, since the diameter of each segment is $O(\sqrt{n})$.

At the end, each tree edge knows the value of $f$ computed on short-range, mid-range and long-range edges that cover it. Computing $f$ on the outputs, results in the final output of $f$. The whole computation takes $O(D+\sqrt{n})$ rounds.    
\end{proof}

\subsection{Decomposing the tree into layers} \label{sec:layers}

Our algorithm uses a decomposition of the tree into layers, that is described also in \cite{agarwal2018set}. We show how to compute the decomposition and work with it in the $\mathsf{CONGEST}$ model. Following the terminology in \cite{agarwal2018set}, we say that a vertex is a \emph{junction} if it has more than one child in the tree.
The first layer is composed of all the tree edges in the tree paths between a leaf to its lowest ancestor that is a junction, or to the root if there is no junction in the tree. After computing the first layer, we can contract all the paths in the first layer, and repeat the process on the contracted graph, now with respect to the leaves and junctions in the contracted graph. This gives us the second layer. We repeat in the same manner until the whole graph is contracted to one vertex. See Figure \ref{layering} for an illustration. We next discuss simple observations about the layering.

\setlength{\intextsep}{0pt}
\begin{figure}[h]
\centering
\setlength{\abovecaptionskip}{-2pt}
\setlength{\belowcaptionskip}{6pt}
\includegraphics[scale=0.5]{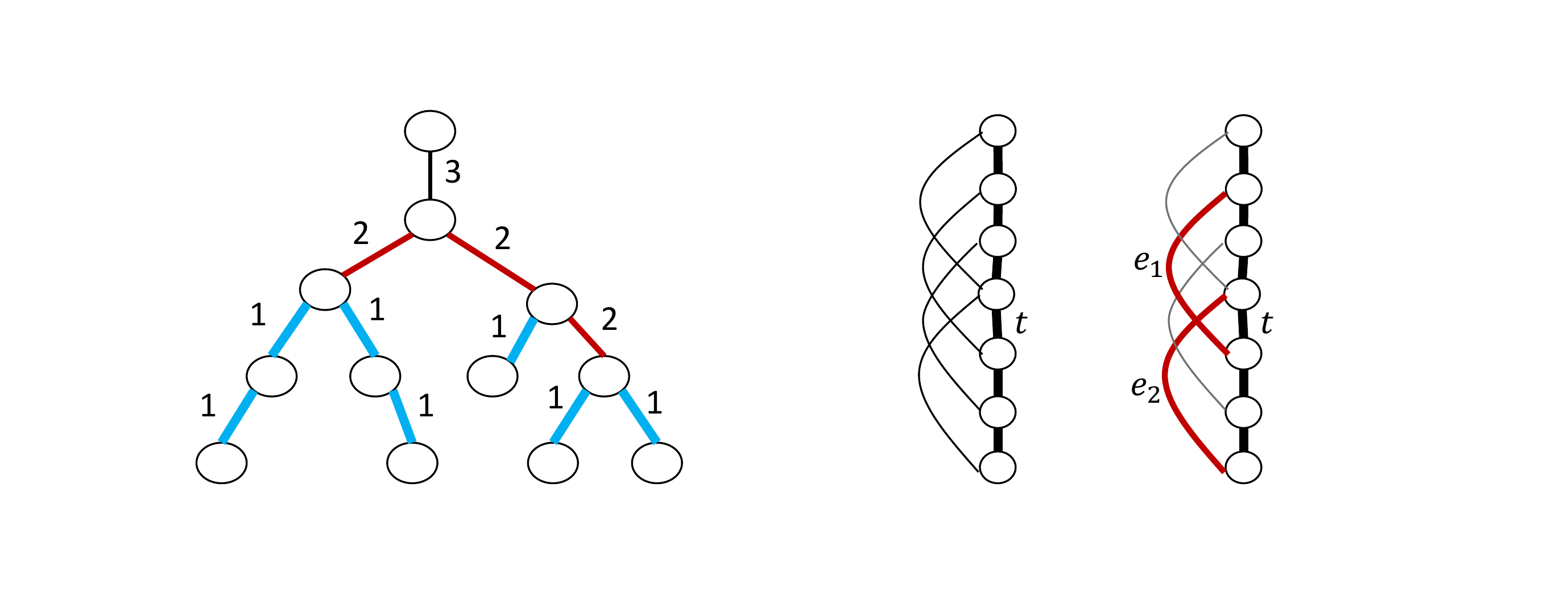}
 \caption{An illustration of the layering. On the left, there is a tree decomposed into layers. On the right, there is a tree (the path with bold edges) and non-tree edges that cover it, with an example of a tree edge $t$ and its two petals $e_1,e_2$.}
\label{layering}
\end{figure}

\begin{claim}
There are $O(\log{n})$ layers in the decomposition.
\end{claim}

\begin{proof}
Each time we contract the graph, the number of leaves decreases by a factor of 2. This happens since any leaf in the contracted graph is an ancestor of at least two leaves in the previous graph. Hence, after at most $O(\log{n})$ iterations we are left with a single path, in one additional iteration the process terminates.
\end{proof}

From the above description, each layer is composed of disjoint paths in the tree, where the first layer includes all the paths between leaves to their first ancestor that is a junction, and so on. We next show that any non-tree edge between an ancestor to its descendant can cover edges in at most one of these paths in each layer.

\begin{claim} \label{claim_intersect}
Let $e$ be an edge between an ancestor to its descendant in the tree. Then, the edges that $e$ covers intersect at most one path in each layer.
\end{claim}

\begin{proof} 
Let $layer(t)$ be the layer number of the tree edge $t$.
If $t_1$ and $t_2$ are in the same leaf to root path in the tree where $t_2$ is closer to the root, it follows that $layer(t_2) \geq layer(t_1)$. Hence, each leaf to root path in the tree is grouped into consecutive tree edges in the same layer, all of them are part of the same path in this layer. Since a non-tree edge between an ancestor to its descendant covers a part of a leaf to root path in the tree, the claim follows. 
\end{proof}

Let $X$ be a subset of non-tree edges, we say that the two tree edges $t$ and $t'$ are \emph{neighbours} with respect to $X$ if there is a non-tree edge in $X$ that covers both $t$ and $t'$.
In our algorithm, all the non-tree edges are between ancestors and descendants, which is crucial for the following section.

\paragraph{The petals of a tree edge.}
The layering is useful in our algorithm for the following reason. For each tree edge $t$, if we look at a subset of non-tree edges that cover it, we can replace all of them by only two non-tree edges that cover all the neighbours $t'$ of $t$ with $layer(t') \geq layer(t)$. These non-tree edges are the \emph{petals} of $t$. We next formalize this notion and explain how the petals are defined. 

Let $X$ be a set of non-tree edges, and let $t$ be a tree edge in layer $i$ covered by $X$, where $p(t) \subseteq X$ is the set of non-tree edges that cover $t$. We next show that we can replace $p(t)$ by two non-tree edges in $p(t)$ that cover all the tree edges covered by $p(t)$ in layers $i,...,L$. We call these two non-tree edges the \emph{petals} of $t$ in $p(t)$.
The petals of $t$ are defined as follows. For simplicity, we start by describing the petals of tree edges $t$ in the first layer. Each layer is composed of disjoint paths. Let $t=\{v,p(v)\}$ be an edge in the first layer in the path $P$. Any non-tree edge that covers $t$ is of the form $\{anc,dec\}$ where $anc$ is an ancestor of $p(v)$, and $dec$ is a descendant of $v$. The edge that gets to the highest ancestor is the first petal of $t$, and the edge that gets to the lowest descendant is the second petal. Note that since $t$ is in the first layer, the subtree rooted at $v$ is just a path, and all the edges that cover $t$ have a vertex which is a descendant of $v$ in this path, hence we can compare them. See Figure \ref{layering} for an illustration. 

For a tree edge $t$ in layer $i$, we define the first petal accordingly, this is just the non-tree edge that gets to the highest ancestor. For the second petal, we cannot just define it as the non-tree edge that gets to the lowest descendant, since not all the non-tree edges that cover $t$ cover the same leaf to root path, as happens in the first layer. To overcome this, we define the second petal as follows. Let $P$ be the tree path in layer $i$ where $t \in P$. We call the lowest vertex in this path, the \emph{leaf} of the path, and denote it by $leaf(P)$. Also, each tree edge is part of a path $P$ is some layer $i$, we denote by $leaf(t)$ the leaf of $P$ where $t \in P$. Let $e=\{anc,dec\}$ be an edge between an ancestor $anc$ to its descendant $dec$ that covers $t=\{v,p(v)\}$. Then, $dec$ is a descendant of $v \in P$, let $u_e=LCA(leaf(t),dec)$. Note that since $dec$ has an ancestor in $P$, and $leaf(t)$ is the lowest vertex in $P$, the vertex $u_e$ is in $P$. For all the non-tree edges $e$ that cover $t$ we can compute the value $u_e$. The second petal is defined to be the non-tree edge $e$ where $u_e$ is the lowest vertex in $P$ among the vertices $u_e$ computed (since all the vertices $u_e$ are in $P$ we can compare them). We next prove that the petals defined really satisfy the required properties.   

\begin{claim}
Let $X$ be a set of non-tree edges, and let $t$ be a tree edge in layer $i$ covered by $X$, where $p(t) \subseteq X$ is the set of non-tree edges that cover $t$. Then, the petals of $t$ cover all the tree edges covered by $p(t)$ in layers $i,...,L$. 
\end{claim}

\begin{proof}
Let $t'$ be a neighbour of $t$ with respect to $X$ with $layer(t') \geq i$, and let $e=\{anc,dec\} \in X$ be a non-tree edge that covers $t$ and $t'$. Since $e$ is between an ancestor to its descendant, $t$ and $t'$ are in the same leaf to root path in the tree. If $t'$ is closer to the root in this path, the higher petal of $t$ covers $t'$, because $e$ is a tree edge that covers $t$ and $t'$, and the higher petal is the edge that covers $t$ and gets to the highest ancestor possible, which is at least $anc$. We are left with the case that $t$ is closer to the root. Hence, $layer(t') \leq layer(t)=i$. Since $layer(t') \geq i$, $t$ and $t'$ are in the same layer and are in the same root to leaf path, which means that $t' = \{v',p(v')\}\in P$, where $P$ is the path in layer $i$ where $t \in P$. Since $e$ covers $t'$ and $t$, there is an edge that covers the subpath between $v'$ to $t$. From the definition of the lower petal, it must cover this subpath, and maybe additional edges in $P$ below it, which completes the proof.
\end{proof}

\subsubsection*{Computing the layers}

We next describe how we compute the layers in $\widetilde{O}(D+\sqrt{n})$ time. All the tree edges $t$ would learn $layer(t)$ and $leaf(t)$, and all the vertices would learn if they are the highest or lowest vertices in a path $P$ in some layer, we call these vertices the root or leaf vertices of $P$. Additionally, we show how all the tree edges can compute their petals with respect to a subset of non-tree edges $X$.

\begin{claim} \label{claim_layer}
In $O((D+\sqrt{n})\log{n})$ rounds, all the tree edges $t$ know $layer(t)$ and $leaf(t)$, and for each $1 \leq i \leq L$, all the vertices know if they are the leaf or root vertices of a path $P$ in the layer $i$.
\end{claim}

\begin{proof}
Our algorithm proceeds in $O(\log{n})$ iterations, according to the layers. We start by describing the first iteration. First, all the vertices learn if there is a junction in the subtree rooted at them. For this, we use the decomposition described in Section \ref{sec:dec}. First, in each segment locally we use a simple convergecast to solve the task: basically we scan the segment from its leaves to the root, where each vertex informs its parent if in the subtree rooted at it there is a junction. This takes $O(\sqrt{n})$ rounds since this is the diameter of segments. Now, all the vertices in the virtual tree, which are ancestors or descendants in a segment, let all the vertices in the graph know if there is a junction in the subtree rooted at them in the segment, or if they are junctions (this is relevant for the case that there is a vertex which is the root of several segments). This takes $O(D+\sqrt{n})$ rounds using Claim \ref{info2}. Also, all the vertices know the complete structure of the skeleton tree in $O(D+\sqrt{n})$ rounds using Claim \ref{info1}.
A vertex $v$ learns if in the subtree $T_v$ rooted at it there is a junction based on the local and global information. If in the local part it learns that there is a junction in the subtree rooted at $v$ in its segment, it already has an answer. Otherwise, from the structure of the skeleton tree $v$ knows which segments are descendants of it (it has such descendants only if it is a vertex on the highway of the segment), and if one of them has a junction in the subtree rooted at it in the segment or is a junction in the skeleton tree. This allows $v$ to learn if there is a junction in $T_v$.
In a similar way, all the vertices learn what is the label of the minimum leaf in the subtree rooted at them (recall that in our algorithm we use LCA labels of vertices instead of ids). This again requires computing an aggregate function from the leaves to the root, that can be computed similarly.

Given this information, we compute the first layer, as follows. A tree edge $t=\{v,p(v)\}$ is in the first layer iff in the subtree rooted at $v$ there is no junction. In such case there is only one leaf in the subtree rooted at $v$ which is the leaf $v$ learns about, and this is $leaf(t)$. The leaves of the first layer are exactly the leaves of the tree, and the roots of a path are all the junctions $p(v)$, where the tree edge $\{v,p(v)\}$ is in the first layer. The leaves of the second layer are all the junctions where all the tree edges from them to their children are in the first layer. Since all the edges know if they are in the first layer, we can contract all the edges of the first layer to obtain a new tree, and then run the same algorithm on this tree to compute the second layer. The new leaves in this tree are the leaves of the second layer and junctions are vertices that at least two of the edges from them to their children are not in the first layer. The vertices and edges that were contracted do not participate in the algorithm. As for the skeleton tree, it may be the case that complete segments are contracted, in this case the root of the segment knows about it and can pass the information to the rest of the graph. However, if a segment is not completely contracted, from the definition of the layers, all the segments above it are not fully contacted as well. Hence, the algorithm works with a contracted version of the original skeleton tree that is known to all the vertices. We run on the contracted graph exactly the same algorithm we run for computing the first layer, which allows us computing the second layer. We continue in the same manner to compute all the $O(\log{n})$ layers, since each iteration takes $O(D+\sqrt{n})$ rounds, the complexity is $O((D+\sqrt{n})\log{n})$ rounds.  
\end{proof}

\begin{claim} \label{claim_petals}
For each $1 \leq i \leq L$, in $O(D+\sqrt{n})$ rounds all the tree edges in layer $i$ learn their petals with respect to a subset of non-tree edges $X$.
\end{claim}

\begin{proof}
The first petal that covers a tree edge $t$ is the non-tree edge $e \in X$ that covers $t$ that gets to the highest ancestor. This is simply an aggregate function of the non-tree edges that cover $t$, which can be computed in $O(D+\sqrt{n})$ rounds using Claim \ref{agg_non_tree}. Note that since we work with the LCA labels of edges, it is easy to compare two non-tree edges that cover $t$, with ancestors $u_1,u_2$. To check which is the highest ancestor we just compute $LCA(u_1,u_2)$. 

To compute the second petal, we work as follows. Using Claim \ref{claim_layer}, all the tree edges $t$ know $layer(t)$ and $leaf(t)$. Let $t$ be a tree edge in layer $i$. For a non-tree edge $e$ that covers $t$, we would like $e$ to learn $leaf(t)$. Notice that by Claim \ref{claim_intersect}, $e$ covers a part of at most one path in layer $i$. Hence, it needs to learn only one value $leaf(t)$ for layer $i$. This is an aggregate function of the tree edges that $e$ covers: it needs to learn about the minimum value $leaf(t)$ for all the tree edges it covers in layer $i$. This takes $O(D+\sqrt{n})$ rounds using Claim \ref{agg_tree}. Now, if $e=\{anc,dec\}$, it computes $u_e=LCA(dec,leaf(t))$, which can be computed easily from the labels of $dec$ and $leaf(t)$. The lower petal is the edge $e$ with the lowest descendant $u_e$. This is an aggregate function of the values $(e,u_e)$ where $e$ are the non-tree edges in $X$ that cover $t$. Again, using the LCA labels we can compute which is the lower descendant between $u_{e_1},u_{e_2}$ where $e_1,e_2$ are two non-tree edges that cover $t$. All the values $u_e$ that are relevant for $t$ are comparable from the definition of the lower petal. Computing an aggregate function takes $O(D+\sqrt{n})$ rounds, which completes the proof.
\end{proof}

\subsection{The forward phase} \label{sec:forward-imp}

As explained in Section \ref{sec:agg}, in $O(D+\sqrt{n})$ rounds, all the non-tree edges simultaneously can learn an aggregate function of the tree edges they cover. In addition, in $O(D+\sqrt{n})$ rounds, all the tree edges simultaneously can learn an aggregate function of the non-tree edges that cover them.  
We now show that using these building blocks we can implement the forward phase in $O((D+\sqrt{n})\frac{\log^2{n}}{\epsilon})$ rounds.

\begin{lemma}
The forward phase takes $O((D+\sqrt{n})\frac{\log^2{n}}{\epsilon})$ rounds.
\end{lemma}

\begin{proof}
In each iteration of the forward phase we work with tree edges from a certain layer $k$. All the vertices would know which layer we process in the current iteration, at the beginning $k=1$. In addition, all the tree edges know their layer and during the algorithm they would know if they are already covered or not. At the beginning, all the tree edges are not covered. We next explain how to implement one iteration, assuming that at the beginning all the tree edges know if they are covered or not (they would learn this information at the end of the previous iteration).

At the beginning of the first iteration in epoch $k$, all the non-tree edges $e$ compute the current value of the dual variable $s(e)=\sum_{t \in S_e} y(t)$. This is clearly an aggregate function of the tree edges they cover, and hence takes $O(D+\sqrt{n})$ rounds by Claim \ref{agg_tree}. Next, the non-tree edges learn the value $|S^k_e|=|R_k \cap S_e|$, which is the number of tree edges they cover that are in layer $k$ and are still not covered. Since all the tree edges know if they are covered and if they are in layer $k$, this is the sum of indicator variables $\sum_{t \in S_e} z_t$, where $z_t=1$ if $t \in R_k$ and $z_t=0$ otherwise, which is clearly an aggregate function of $S_e$. Next, each tree edge in $R_k$, sets $y(t) = min_{t \in S_e} \frac{w(e)-s(e)}{|S^k_e|}$. Note that all the non-tree edges $e$, know $w(e)$ and computed $s(e), |S^k_e|$. Hence, the tree edges need to learn an aggregate function of the non-tree edges that cover them, the minimum value $\min_{t \in S_e} \frac{w(e)-s(e)}{|S^k_e|}$. Then, we add to the augmentation $A$ edges $e$ if their dual constraint becomes tight, which is done by computing again $s(e)=\sum_{t \in S_e} y(t)$, and checking whether $s(e) \geq w(e).$ Now, all the tree edges should learn if they are covered or not. This is again an aggregate function of the non-tree edges that cover them: they need to learn if at least one of these edges is in the set $A$, which can be done for example by summing the number of edges in $A$ from the edges that cover each tree edge. After each edge knows if it is covered or not, all the vertices can learn in $O(D)$ rounds if there is at least one edge in layer $k$ that is still not covered by communicating over a BFS tree. If all the edges in layer $k$ are covered, we move to the next layer, and otherwise we continue epoch $k$.
In this case, in the next iteration, all the tree edges in layer $k$ that are still not covered increase their dual variable by a multiplicative factor of $(1+\epsilon)$, which is a completely local task. Again, we check which dual constraints become tight to know which non-tree edges are added to $A$, then each tree edge learns if it is covered, and all vertices learn if there is still an edge in layer $k$ that is not covered. We continue in the same manner until all the edges in layer $k$ are covered.

Since each iteration is based on computing a constant number of aggregate functions, each iteration takes $O(D+\sqrt{n})$ rounds using Claims \ref{agg_tree} and \ref{agg_non_tree}. The number of layers is $O(\log{n})$, and in each layer there are at most $O(\frac{\log{n}}{\epsilon})$ iterations, as follows. Let $t \in R_k$, we show that after at most $O(\frac{\log{n}}{\epsilon})$ iterations $t$ is covered. At the first iteration in epoch $k$, we set $y(t) = min_{t \in S_e} \frac{w(e)-s(e)}{|S^k_e|}$, and in each iteration we increase $y(t)$ by a multiplicative factor of $(1+\epsilon)$. Let $y_0$ be the initial value of $y(t)$. Note that if we reach a stage where $y(t) = y_0 \cdot |S^k_e|$, the dual constraint of the edge $e$ becomes tight, since $y(t)$ contributes $w(e)-s(e)$, where $s(e)$ was the value of the dual constraint before epoch $k$. Now, since $|S^k_e| \leq n$, and in the $i$'th iteration of epoch $k$, $y(t)=(1+\epsilon)^{i}y_0$ (unless $t$ is already covered), after $O(\frac{\log{n}}{\epsilon})$ iterations the dual constraint of $e$ becomes tight, and it is added to $A$, unless $t$ is already covered before. To conclude, the complexity of the forward phase is $O((D+\sqrt{n})\frac{\log^2{n}}{\epsilon})$ rounds.
\end{proof}

\subsection{The reverse-delete phase} 

\subsubsection{Implementation details} \label{sec:rev-del-imp}

In Section \ref{sec:rev-del}, we showed that the number of iterations in the reverse-delete phase is $O(\log^2{n})$. We next show how to implement each iteration in $O(D+\sqrt{n})$ rounds. First, note that from the forward phase tree edges know in which epoch they are covered, and non-tree edges in $A$ know when they were added to $A$. Also, all the tree edges know their layer number. From the above all the edges can compute if they are in $R_k,A_k$, and $F_k$. $R_k$ are edges in layer $k$ that are first covered in epoch $k$, $A_k$ are non-tree edges that are added to $A$ in epoch $k$, and $F_k$ are the tree edges that are first covered in epoch $k$. Also, in each iteration edges know if they are in the sets $X,Y,B,H_i$ and $\widetilde{H}_i$. At the beginning $B=\emptyset$ and $X=B \cup A_L$. Later on, we add edges to $B$, and all the edges added to $B$ or $Y$ learn about it in the algorithm, which allows computing $X=B \cup A_k$ in the next iterations. In epoch $k$, the set $H_i$ are all the edges in layer $i$ in $F=\cup_{i=k}^{L} F_i$, which can be computed since edges know their layer number and whether they are in $F$. During the algorithm, tree edges learn if they are already covered by edges added to $Y$, which allows computing $\widetilde{H}_i$, the tree edges that are still not covered in $H_i$. Also, from Claim \ref{claim_petals}, in $O(D+\sqrt{n})$ rounds all the tree edges in layer $i$ can compute their petals with respect to a subset of non-tree edges $X$. 

Hence, to compute one iteration we need to explain how to build an MIS in the virtual graph $G_i$, how all the non-tree edges learn if they are added to $Y$, and how non-tree edges learn if they are already covered by $Y$. At the first iteration in epoch $k$, $Y=\emptyset$, hence all the tree edges are not covered by $Y$.

\subsubsection*{Computing an MIS in the graph $G_i$}
We next explain how to implement one iteration. At this point, all the tree edges know if they are in $\widetilde{H}_i$, and we would like to find an MIS in $G_i$. I.e., we would like to find a maximal set of tree edges $M_i$ in $\widetilde{H}_i$ such that for any two edges $t,t' \in M_i$, there is no edge in $X$ that covers both of them. We design a specific algorithm to solve this task. For this, we use the decomposition of the tree into segments presented in Section \ref{sec:dec}, and exploit the structure of the layers. As explained in Section \ref{sec:dec}, we can decompose the tree into $O(\sqrt{n})$ segments of diameter $O(\sqrt{n})$, that each of them has a main path called the \emph{highway} of the path, and additional subtrees attached to it that are not connected to the rest of the tree. Also, all the vertices know the structure of the decomposition, captured by the \emph{skeleton tree}.  

As explained in Section \ref{sec:layers}, each layer is composed of disjoint paths, that are not in the same root to leaf path in the tree. In particular, each highway intersects at most one path in each layer.
A path $P$ in layer $i$ is between a descendant to its ancestor in the tree, where the descendant is denoted by $leaf(P)$. The path $P$ starts with a subpath in the segment of $leaf(P)$, and possibly parts of highways above this segment, since the only ancestors of $leaf(P)$ in other segments are highway vertices.  
For each highway that participates in such path in the layer $i$, let $t_h,t_{\ell}$ be the highest and lowest edges of $P$ in the highway that are still not covered at the beginning of the iteration (if such exist). For each segment, we let all the vertices in the graph learn these edges and their petals with respect to $X$, as follows. First, by Claim \ref{claim_petals}, all the tree edges in layer $i$ can learn their petals in $O(D+\sqrt{n})$ rounds. Then, by scanning the highway, each segment can learn the edges $t_h,t_{\ell}$ and their petals in $O(\sqrt{n})$ rounds. Now all the vertices can learn all the information in $O(D+\sqrt{n})$ rounds using Claim \ref{info2}, since they need to learn $O(\log{n})$ information per segment.

Let $T'$ be all the edges $t_h,t_{\ell}$ of all the segments. Our MIS algorithm consists of a global part where we compute an MIS of the edges $T'$, and a local part where we scan all the local subpaths in layer $i$ in each segment, and compute an MIS for them. We next describe the algorithm in detail and prove its correctness.\\[-7pt] 

\textbf{Computing a global MIS.} Since all the vertices know the edges in $T'$ and their petals, each vertex can compute locally the virtual graph $G'$ that its vertices are $T'$ and two tree edge $t,t'$ are connected iff there is an edge in $X$ that covers both $t,t'$. Knowing the petals of each tree edge is enough for this task, and vertices do not need to know the whole set $X$. This holds since the petals cover all the neighbours of $t$ in the same layer or above, and in iteration $i$ we consider only the tree edges in layer $i$. Also, from the labels of the petals, vertices know which of the tree edges are covered by them, as explained in Observation \ref{obs_labels} (our whole algorithm works with the LCA labels, instead of ids, as mentioned earlier). Given the graph $G'$, all the vertices can compute an MIS, $M'$, in this graph by applying the same algorithm. For example, the sequential greedy algorithm for MIS. Then, all the petals of edges in $M'$ are added to the cover $Y$. All the vertices compute this information locally, without communication. At the end, all the non-tree edges know if they were added to $Y$, which happens only if they are among the petals added. Also, all the tree edges know if they are covered or not, since they know which edges were added to $Y$, and given the LCA labels of a non-tree edge, we can compute easily if it covers a tree edge by Observation \ref{obs_labels}.\\[-7pt]
 
\textbf{Computing local MISs.} Then, in each segment locally, we compute an MIS for all the tree edges in $\widetilde{H}_i$ that are still not covered, this is done as follows. Let $P$ be a path in layer $i$, in each segment we look only at the part of $P$ in the segment if exist (there may be several such paths in the segment). In each one of them, we simulate a greedy MIS algorithm, as follows. We start with the lower vertex in the path $v$, that may be either $leaf(P)$, or the unique descendant of a segment. Note that all the vertices know if they are the leaf or root of a path in layer $i$ by Claim \ref{claim_layer}. The vertex $v$ checks if the edge to its parent $e=\{v,p(v)\}$ is already covered by $Y$, if not it is added to the MIS. Also, the vertex $v$ knows the higher petal of $e$, denote it by $\{anc,dec\}$ where $anc$ is an ancestor of $dec$, and it sends $anc$ to its parent, for the following reason. Since $e$ is added to the MIS, its petals would be added later to $Y$. In particular, all the path between $v$ to $anc$ is going to be covered by the edge $\{anc,dec\}$, and hence we do not add other tree edges in this path to the MIS. The reason it is enough to focus on the higher petal is since we scan the paths from a descendant to its ancestor, and the lower petal only covers additional tree edges below $e$, which are not relevant. We continue in the same manner, now the vertex $v'=p(v)$ checks if the edge $t'=\{v',p(v')\}$ is already covered either by the original set $Y$ or by edges added to $Y$ by vertices below it in the segment. For the latter task, it just checks if the vertex $anc$ it receives from $v$ is an ancestor of $p(v')$ or not, which can be deduced from the LCA labels. If $t'$ is already covered, $v'$ sends to its parent the label of the highest ancestor in an edge that covers $t'$, from edges added to $Y$ by vertices below it. Otherwise, the edge $t'=\{v',p(v')\}$ is added to the MIS and $v'$ adds its higher petal to cover $t'$, and sends to its parent the ancestor of its higher petal (which is again the highest ancestor in an edge that covers $t'$ from the edges added to $Y$). We continue until we reach either the highest vertex in the path $P$ or the ancestor of the segment. At the end, we add all the petals of tree edges in all the MISs computed in all the segments to the cover $Y$.\\[-7pt]

At the end of the process, all the tree edges in $H_i$ are covered. If they were covered already by $Y$ from previous iterations or after computing the global MIS, it is clear. Otherwise, we reach each tree edge $t \in H_i$ during the local computations, and if it is not covered when we reach it, we add it to the MIS and add its petals to $Y$, which cover it.    

Let $M_i$ be all the edges added to the MIS, either when computing a global MIS, or in one of the local computations in the segments. We next show that although we work in different segments in parallel, the set $M_i$ really forms an MIS in the graph $G_i$.

\begin{claim} \label{claim_mis}
The set $M_i$ is an MIS in the graph $G_i$.
\end{claim}

\begin{proof}
Let $t_1,t_2 \in M_i$ be two tree edges added to the MIS. Both of them must be added in the local parts, since the edges added in the global part form an MIS, $M'$, by their definition. Also, after the global part we add all the petals of edges in $M'$ to the set $Y$, which cover all the neighbours of edges in $M'$. Hence, all the tree edges in $H_i$ that are still not covered, do not have a neighbour in $M'$. In addition, $t_1$ and $t_2$ cannot be in the same segment, as follows. First, note that $t_1$ and $t_2$ can only be neighbours in $G_i$ if they are in the same leaf to root path in the tree, since all the non-tree edges are between ancestors to descendants. Assume that $t_1,t_2$  are neighbours in the same segment, and assume w.l.o.g that $t_1$ is below $t_2$ in the tree. Then, when we reach $t_1$, since it is added to the MIS, we add its higher petal to the MIS which must cover $t_2$ if they are neighbours where $t_2$ is above $t_1$. Then, when we reach $t_2$ it is already covered and is not added to the MIS. 

We are left with the case that $t_1,t_2$ are added in the local computations of different segments.  
Let $P'$ be the tree path between $t_1$ to $t_2$, assume that $t_2$ is closer to the root in this path, and let $t$ be the lowest edge in the highway of $t_2$ that is not covered at the beginning of iteration $i$. Note that from the structure of the segments, $t_2$ must be part of the highway of its segment. Hence, there must exist such an edge $t$, because $t_2$ is an highway edge that is not covered at the beginning of the iteration. Also, since $t_1$ and $t_2$ are neighbours there is a non-tree edge $e \in X$ that covers both of them. Since $t$ is in the path $P'$, is follows that $e$ covers $t$ as well, which shows that $t$ is a neighbour of $t_1$ and $t_2$. Hence, the petals of $t$ cover $t_1$ and $t_2$. If $t$ is added to the global MIS, $t_1,t_2$ are already covered by these petals, which gives a contradiction. Otherwise, $t$ has a neighbour added to the global MIS. If this neighbour is above $t$, its petals (that cover $t$) must cover also $t_2$, otherwise its petals must cover $t_1$, either way gives a contradiction.
\end{proof}

The computation of the global MIS required learning the petals of the tree edges in layer $i$ and learning $O(\log{n})$ information per segment, which takes $O(D+\sqrt{n})$ rounds, and the computation of the local MISs requires scanning each segment which takes $O(\sqrt{n})$ rounds.\\[-7pt] 

\textbf{Completing the iteration.}
At the end of each iteration, all the non-tree edges should learn if they were added to $Y$, and all the tree edges should learn if they are currently covered by the set $Y$. In the global MIS part this is easy since all the vertices simulate the algorithm locally and know which edges are added to $Y$, as explained above. We next explain how this is done for the local part. First, after each non-tree edge knows if it is in $Y$, each tree edge can learn if it is covered by $Y$, since this is an aggregate function of edges that cover it, this takes $O(D+\sqrt{n})$ rounds as explained also in the forward phase. 

In order that non-tree edges would learn if they were added to $Y$, we work as follows. The edges added to $Y$ are all the petals (with respect to $X$) of the edges in $M_i$, all these petals are in $X$. Hence, if a non-tree edge $e \in X$ is added to $Y$, it is one of the petals of a tree edge $t \in M_i$. All the tree edges know their petals and know if they are in $M_i$ from the algorithm. From the independence of $M_i$, there are no two tree edges $t_1,t_2 \in M_i$ that are covered by the same edge from $X$. Hence, for each non-tree edge $e \in X$, there is at most one edge $t \in M_i \cap S_e$. Therefore, it is enough that all the non-tree edges would learn which is the first edge $t \in M_i \cap S_e$ (according to some order), and what are its petals (information of $O(\log{n})$ bits), since this is an aggregate function, by Claim \ref{agg_tree}, all the non-tree edges can learn it in $O(D+\sqrt{n})$ rounds, and learn if they are in $Y$.

To conclude, we can implement an iteration in $O(D+\sqrt{n})$ rounds. Since the reverse-delete phase has $O(\log^2{n})$ iterations, the time complexity of the whole phase is $O((D+\sqrt{n})\log^2{n})$ rounds.

\begin{lemma}
The reverse-delete phase takes $O((D+\sqrt{n})\log^2{n})$ rounds.
\end{lemma}

\subsection{An improved approximation} \label{sec:improved}

We next present a variant of our 2-ECSS algorithm that obtains an improved approximation of $(5+\epsilon)$ instead of $(9+\epsilon)$. The approximation for weighted TAP is $(4+\epsilon)$, which almost matches an $\widetilde{O}(D+\sqrt{n})$-round 4-approximation for the unweighed case \cite{DBLP:conf/opodis/Censor-HillelD17}. The difference is in the reverse-delete phase. In the algorithm presented in Section \ref{sec:rev-del}, we make sure that each tree edge $t$ with $y(t) > 0$ is covered at most 4 times by $B$. We now construct the set $B$ slightly differently in a way that guarantees that each tree edge $t$ with $y(t) > 0$ is covered at most 2 times by $B$, which results in an improved approximation for the whole algorithm.

All the definitions and the general structure of the algorithm follows the algorithm described in Section \ref{sec:rev-del}. The only difference is that in epoch $k$ we make sure that the edges in $R_k$ are covered at most 2 times instead of 4 times, which affects the computation in iteration $i$.
We next describe the difference in the reverse-delete phase. We start by describing a variant of the algorithm that guarantees that tree edges in $R_k$ are covered at most 3 times, and then show how to improve this to 2.

\subsubsection*{Covering $R_k$ at most 3 times}

We again build a set $B \subseteq A$, where initially $B=\emptyset$, and we proceed in epochs $k=L,...,1$, where in epoch $k$ we make sure that all the edges in $F_i$ for $i \geq k$ are covered by $B$, and now all the edges in $R_i$ for $i \geq k$ are covered at most 3 times. In epoch $k$, we again define $X = B \cup A_k$ (with $B$ at the end of epoch $k+1$) and $F = \cup_{i=k}^{L} F_i$, and build a set $Y \subseteq X$ that covers $F$, where at the end of epoch $k$ we set $B=Y$.
We go over the layers in iterations $i=k,...,L$, where in iteration $i$, we again want to cover all the edges in $H_i$, which are all the edges in layer $i$ in $F$. $\widetilde{H}_i$ are all the tree edges in $H_i$ that are not covered by $Y$ at the beginning of iteration $i$.

\paragraph{Iteration $i$.}
In Section \ref{sec:rev-del}, we build an MIS, $M_i$, of the edges $\widetilde{H}_i$, and then add their petals to $Y$. Here, we take a different approach, we choose a subset $M'_i$ of edges in $\widetilde{H}_i$ that are not necessarily independent, but have a certain structure that guarantees that there are only small number of dependencies between them. For each edge in $M'_i$ we add only its higher petal to $Y$, and the edges are chosen in such a way that all edges in $\widetilde{H}_i$ are covered by $Y$, and each tree edge in $R_k$ is covered at most 3 times (where $k$ is the epoch number). 

The algorithm starts with a \emph{global part}, in which we compute a global MIS, $M'$, exactly as described in Section \ref{sec:rev-del-imp}, but now for each one of the edges in $M'$ we just add its higher petal to $Y$. Then, we have a \emph{local part} where we compute local MISs in each segment separately, with the goal of covering all the tree edges of $H_i$ in the segment that are still not covered. The computation is identical to the computation of local MISs in Section \ref{sec:rev-del-imp}, with the difference that for each edge added to the local MIS we only add its higher petal to $Y$. We denote by $M'_i$ all the edges added either to the global MIS or to one of the local MISs, and call them \emph{anchors} as before. We say that an anchor is a \emph{global} anchor if it is added in the global part, and a \emph{local} anchor if it is added in the local part. The difference from before is that the whole set $M'_i$ is no longer independent, because in the global part we do not add the two petals of each anchor, but only the higher one. 
However, we can show that the dependencies have a certain structure, as follows. See Figure \ref{depend_pic} for an illustration.

\setlength{\intextsep}{0pt}
\begin{figure}[h]
\centering
\setlength{\abovecaptionskip}{-2pt}
\setlength{\belowcaptionskip}{6pt}
\includegraphics[scale=0.6]{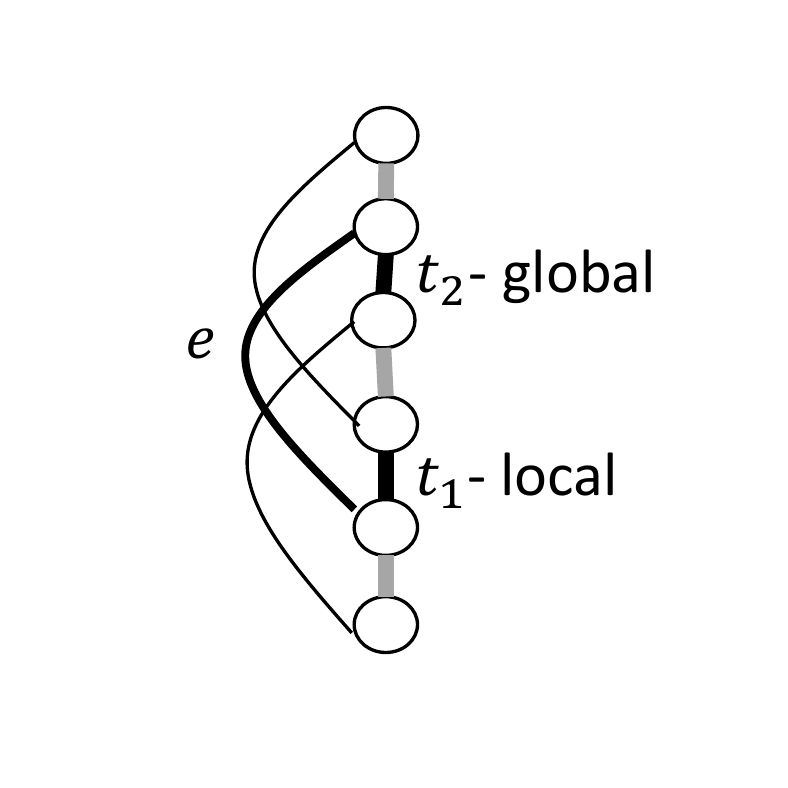}
 \caption{\small An illustration of the dependencies between anchors. Here, the two anchors $t_1$ and $t_2$ have a non-tree edge $e$ that covers them.
\vspace{-5pt}}
\label{depend_pic}
\end{figure}

\begin{claim} \label{anchors}
Assume that the two anchors $t_1,t_2$ have a non-tree edge from $X$ that covers them, and that $t_2$ is closer to the root, then $t_2$ is a global anchor and $t_1$ is a local anchor. In addition, $t_1$ and $t_2$ are added to $Y$ in the same iteration $i$. 
\end{claim} 

\begin{proof}
If $t_1,t_2$ are covered by an edge $e \in X$, they must be in the same root to leaf path. Since $t_2$ is closer to the root in this path, $layer(t_2) \geq layer(t_1)$. Note that since $e$ covers $t_1$ and $t_2$, it follows that the higher petal of $t_1$ covers $t_2$. If $t_1,t_2$ are not in the same layer, then $t_2$ is already covered at the end of iteration $i=layer(t_1)$, a contradiction to its definition. Also, all the anchors added to $M'$ in the global part of iteration $i$ are independent. Hence, at least one of $t_1,t_2$ is added in the local part. If $t_1$ is added in the global part, and $t_2$ is added in the local part, then since the higher petal of $t_1$ covers $t_2$, $t_2$ cannot be an anchor, a contradiction. Hence, $t_1$ is added in the local part. To complete the proof, we need to show that $t_2$ is not added in the local part. Assume to the contrary that it is added in the local part. First, $t_1$ and $t_2$ cannot be in the same segment, since all the tree edges added to $M'_i$ in the local part of the same segment are independent by the algorithm, as explained also in the proof of Claim \ref{claim_mis} (we follow exactly the same algorithm for computing a local MIS in each segment). 

We are left with the case that $t_1$ and $t_2$ are added in the local part in different segments. Let $t_{\ell}$ be the lowest highway edge in the segment of $t_2$ that is in $\widetilde{H}_i$ for $i=layer(t_1)=layer(t_2)$. Note that such an edge must exist and is in the path between $t_1$ and $t_2$, since $t_2 \in \widetilde{H}_i$ is an highway edge from the structure of the segments. 
If $t_{\ell}$ is an anchor, its higher petal covers $t_2$ (since the edge $e$ is an edge that covers $t_1$ and $t_2$, and $t_{\ell}$ is in the path between $t_1$ to $t_2$), a contradiction. Otherwise, there is a global anchor $t_{anc} \in \widetilde{H}_i$ that its higher petal $e'$ covers $t_{\ell}$. If $t_{anc}$ is below $t_1$, then $e'$ must cover $t_1$ as well. If $t_{anc}$ is above $t_2$, then $e'$ must cover $t_2$ as well. Otherwise, $t_{anc}$ is in the path between $t_1$ and $t_2$. Since the higher petal of $t_1$ covers $t_2$, and $t_{anc}$ is above $t_1$, it follows that $e'$ covers $t_2$. Since $e'$ is added in the global part, all the cases lead to a contradiction. 
\end{proof}

\begin{claim} \label{3times}
After the end of epoch $k$, each tree edge $t \in R_k$ is covered at most 3 times, and all the edges in $F = \cup_{i=k}^{L} H_i$ are covered by $Y$.
\end{claim}

\begin{proof}
We say that an anchor covers $t$ if its higher petal covers $t$. 
Let $t \in R_k$.
First note that from Claim \ref{anchors}, each anchor has at most one neighbouring anchor, as follows.
If $t$ is a local anchor in layer $i$, it can only have a neighbouring anchor which is a global anchor in layer $i$ and is above $t$. Since the higher petal of $t$ covers all its neighbours above it, and global anchors are independent, there can be at most one such anchor. A similar argument works if $t$ is a global anchor.  This in particular shows that each anchor $t$ is covered at most twice, by the higher petal of $t$ and of its neighbour if exists.

We now analyze the case that $t$ is not an anchor.
First, the lower petal of $t \in R_k$ covers all its neighbours in layer $k$. Note that all the anchors we add in epoch $k$ are in layer $i \geq k$.
Hence there are at most two anchors below $t$ that cover it, they must be in layer $k$ and by Claim \ref{anchors} they must be of the form $a_1,a_2$ where $a_1$ is a local anchor and $a_2$ is a global anchor, and $a_1$ is below $a_2$. Second, the higher petal of $t$ covers all the edges above it, again there are at most two anchors above $t$ that cover it $b_1,b_2$, where $b_1$ is below $b_2$, $b_1$ is a local anchor and $b_2$ is a global anchor, and $b_1,b_2$ are in the same layer. However, if there is a global anchor $b_2$ above $t$ that covers $t$, then it covers the whole path from $t$ to $b_2$. Hence, there are no local anchors in $layer(b_2)$ added in this path. This shows that there is at most one anchor above $t$ that covers $t$, either a local anchor or a global anchor. In all the cases, there are at most 3 anchors that cover $t$, each of them adds only its higher petal to $Y$, which shows that $t$ is covered at most 3 times.

We now show that all the edges in $F = \cup_{i=k}^{L} H_i$ are covered by $Y$ by the end of epoch $k$. This holds since after iteration $i$ all the edges in $H_i$ are covered by $Y$. If they are not in $\widetilde{H}_i$ they are covered before iteration $i$. Otherwise they either are covered by a global anchor, or by a local anchor from the description of the algorithm: in the local part we scan the paths of layer $i$ in each segment and make sure that all the tree edges in $\widetilde{H}_i$ that are still not covered by global anchors get covered.
\end{proof}

\subsubsection*{Covering $R_k$ at most 2 times}

We now would like to add a \emph{cleaning} phase after the end of epoch $k$, that guarantees that each $t \in R_k$ is covered only at most twice, and also all the edges in $F$ are still covered. For doing so, we show that in the cases a tree edge is covered 3 times, we can actually remove one of these edges without affecting the covering of the rest of tree edges.
From the proof of Claim \ref{3times}, it follows that the only cases that a tree edge $t \in R_k$ is covered 3 times have a certain structure, illustrated in Figure \ref{3times_pic}. First, there are exactly 3 anchors $t_1,t_2,t_3$ that cover $t$, all of them are in the same path between a descendant to an ancestor that starts with the path in layer $k$ that $t$ belongs to, where $t_1$ is the lowest edge and $t_3$ is the highest edge. Now $t_1,t_2$ are in layer $k$ below $t$, where $t_1$ is a local anchor and $t_2$ is a global anchor. $t_3$ is either a local or a global anchor above $t$.

\setlength{\intextsep}{0pt}
\begin{figure}[h]
\centering
\setlength{\abovecaptionskip}{-15pt}
\setlength{\belowcaptionskip}{8pt}
\includegraphics[scale=0.6]{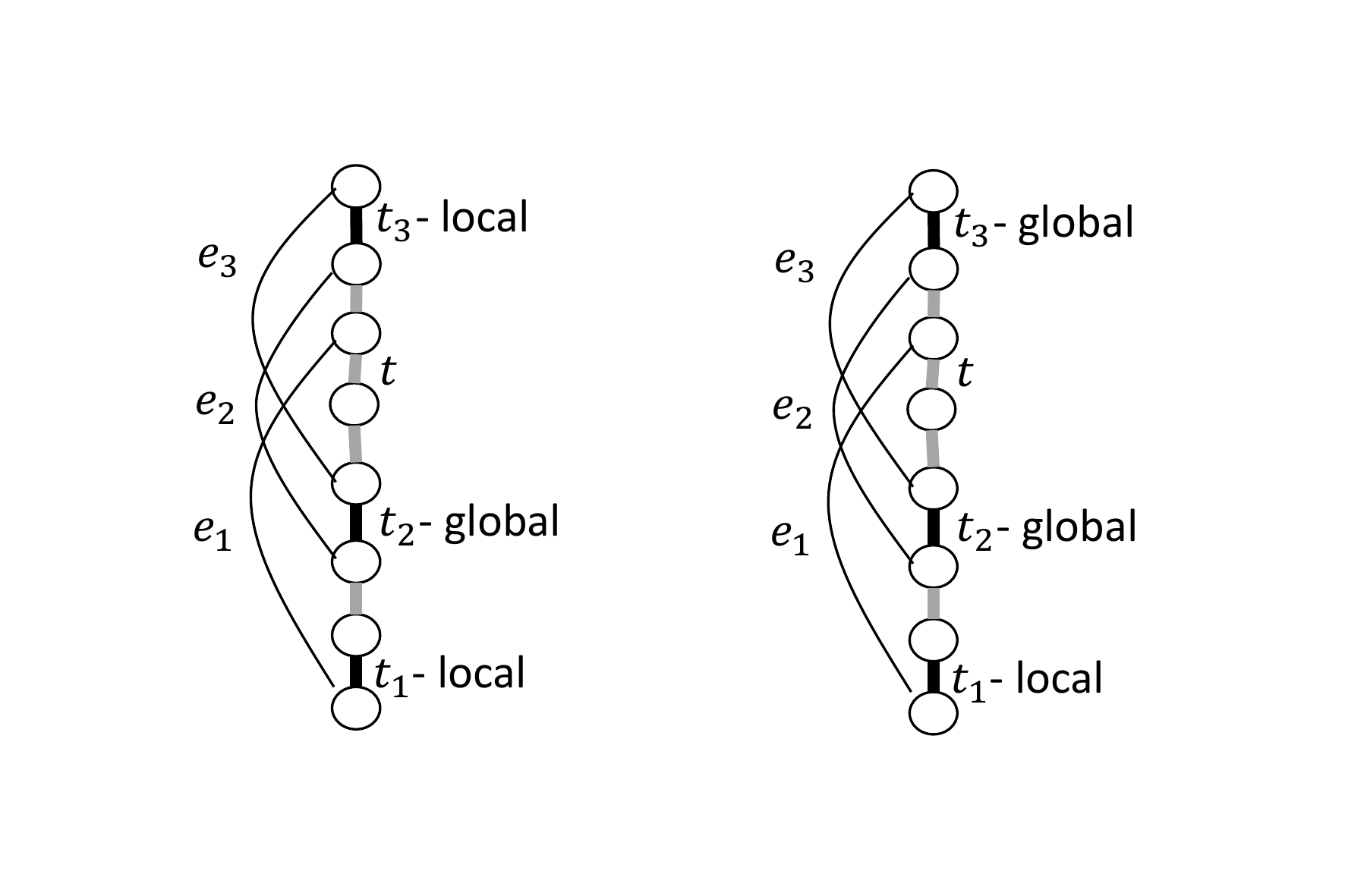}
 \caption{\small An illustration of the cases that a tree edge $t$ is covered 3 times.\vspace{-5pt}}
\label{3times_pic}
\end{figure}

In the cleaning phase, we remove the higher petal of the global anchor $t_2$ from $Y$. We do so for all the tree edges $t \in R_k$ that are covered 3 times. 
After this, all the tree edges $t \in R_k$ are clearly covered at most twice. We next show the following.

\begin{claim} \label{2times}
After the cleaning phase of epoch $k$, each tree edge $t \in R_k$ is covered at most twice, and all the edges in $F$ are covered by $Y$. 
\end{claim}

\begin{proof}
The fact each tree edge $t \in R_k$ is covered at most twice is immediate from the description of the cleaning phase, we now show that all the tree edges in $F$ are still covered. Assume that $t \in R_k$ is a tree edge covered by the 3 anchors $t_1,t_2,t_3$, sorted from the lowest to the highest in the tree. Let $e_1, e_2,e_3$ be the higher petals of these anchors. In the cleaning phase, we removed $e_2$ from $Y$. First, we show that all the tree edges from $F$ covered by $e_2$ are covered either by $e_1$ or $e_3$, this clearly holds for $t$. First, note that all the tree edges from $F$ covered by $e_2$ are in the tree path $P$ that starts in the path of layer $k$ that $t$ belongs to and goes to the root. Let $t' \in F$ be a tree edge covered by $e_2$. If $t'$ is above $t$ in $P$, then $e_3$ covers $t'$. This follows since $e_3$ covers $t$ and the higher petal of $t$ covers all the edges above $t$ covered by $e_2$. Since $t_3$ is above $t$, its higher petal $e_3$, that covers $t$, also covers all the tree edges above $t$ that are covered by $e_2$. If $t'$ is below $t$ in $P$, then it is covered by $e_1$ as follows. Note that the tree edge $t_1$ is in $P$ and it is not covered by $e_2$, since otherwise it would be covered after the global part of iteration $k$, and hence cannot be a local anchor. Since $e_1$ covers $t$, it covers the whole subpath in $P$ between $t_1$ and $t$, which in particular contains all the tree edges in $P$ below $t$ covered by $e_2$ (which can only be above $t_1$). To complete the proof, we show that $e_1$ and $e_3$ remain in $Y$ after the cleaning phase. The edge $e_1$ clearly stays in $Y$ since $t_1$ is a local anchor, and we only remove petals of global anchors in the cleaning phase. However, $t_3$ can be either a local or a global anchor. In the first case, $e_3$ clearly stays in $Y$. We next analyze the case that $t_3$ is a global anchor. First, since $e_3$ covers $t$, in the whole path between $t$ to $t_3$ there are no local anchors: there are no local anchors of layer $i=layer(t_3)$ since they are already covered by $e_3$ after the global part of iteration $i$, and there are no local anchors from previous iterations since otherwise $t_3$ is already covered and is not part of $\widetilde{H}_i$. Also, all the tree path between $t_2$ to $t$ is covered by $e_2$ after the global part of iteration $k$ (which is the first iteration of epoch $k$), so there are no local anchors in the whole path between $t_2$ and $t_3$. Now, if $t_3$ is removed from $Y$, it means that there is a tree edge $\tilde{t} \in R_k$ with 3 anchors that cover it, and in particular two anchors below it $\tilde{t_1},\tilde{t_2}$ where $\tilde{t_2}=t_3$ is a global anchor and $\tilde{t_1}$ is a local anchor below $t_3$. However, as we explained, there are no local anchors between $t_2$ and $t_3$. Also, since all the anchors in epoch $k$ are in layer at least $k$, it means that the only local anchor below $\tilde{t}$ that covers it must be below $t_2$. However, in this case, the lower petal of $\tilde{t}$ covers the global anchors $t_2$ and $t_3$, in contradiction to the fact that global anchors are independent. Hence, the edges $e_1$ and $e_3$ stay in $Y$, which shows that all tree edges in $F$ are covered by the end of the cleaning phase of epoch $k$.  
\end{proof} 

Using Claim \ref{2times}, we get the following.

\begin{lemma} \label{improved_correct_lemma}
At the end of the cleaning phase of epoch $k$:
\begin{enumerate}
\item All the tree edges in $F = \cup_{i=k}^{L} F_i$ are covered by $B$.\label{improved_cover_F} 
\item All the edges in $R_i$ for $i \geq k$ are covered at most 2 times.\label{improved_cover_R} 
\end{enumerate}
\end{lemma}

\begin{proof}
First, \ref{improved_cover_F} follows directly from Claim \ref{2times} ($B = Y$ at the end of the epoch). 
We prove \ref{improved_cover_R} by induction on $k$, the proof follows ideas from Lemma \ref{correct_lemma}. In the base case, $k=L+1$ and $B=\emptyset$, hence the claim clearly holds. We assume that the claim holds for $k+1$, and show that it works for $k$.
We next show that any tree edge $t \in R_i$ for $i \geq k$ is covered at most 2 times by $Y$. For $i>k$ this already follows from previous epochs. This holds since $Y \subseteq X$, and $X=B \cup A_k$ (with the set $B$ at the end of the previous epoch). Now, at the end of epoch $k+1$, all the edges in $R_i$ for $i>k$ are covered at most 2 times by $B$ from the induction hypothesis. Also, from the definition of the sets, $A_k$ does not cover any edge in $R_i$ for $i>k$. Hence, any set $Y \subseteq X$ we choose covers all the edges in $R_i$ for $i>k$ at most 2 times. Now all the edges $t \in R_k$ are covered at most twice by the end of the cleaning phase of the epoch by Claim \ref{2times}. This completes the proof.
\end{proof}

From Lemma \ref{improved_correct_lemma} with respect to $k=1$, we get that at the end of the reverse-delete phase $B$ covers all tree edges, and for each $1 \leq i \leq L$, all the tree edges in $R_i$ are covered at most 2 times, as needed. The algorithm has $L$ epochs, each of them consists of at most $L$ iterations, which sums up to $O(L^2)=O(\log^2{n})$ iterations. We next show how to implement each iteration in $O(D+\sqrt{n})$ rounds, which results in a complexity of $O((D+\sqrt{n})\log^2{n})$ rounds for the whole phase. 

\subsubsection{Implementation details}

The implementation of the algorithm follows the implementation in Section \ref{sec:rev-del-imp} with slight changes, we just highlight the differences. The computation of the global MIS and local MISs is identical to the computation in Section \ref{sec:rev-del-imp}. The only difference is that for each anchor we just add its higher petal and not both its petals (each tree edge knows which is its higher petal). At the end of each iteration all the non-tree edges should know if they are added to $Y$. In Section \ref{sec:rev-del-imp}, we explained that all the edges know which are the global anchors added in the global part and their petals. For the local part we used the independence of anchors. Now although not all the anchors are necessarily independent, it follows from Claim \ref{anchors} that all the local anchors are independent (two anchors can only be neighbours if one is a global anchor and the second is a local anchor). This shows that for each non-tree edge $e \in X$ there is at most one local anchor $t \in M'_i \cap S_e$, hence we can follow the implementation in Section \ref{sec:rev-del-imp}.

In the new algorithm, we also have a cleaning phase after each epoch. In the cleaning phase, all the edges $t \in R_k$ that are covered 3 times remove one of these edges from $Y$. The edge removed is the higher petal of a global anchor $t_2$ which is below $t$. Also, from the analysis there is exactly one global anchor below $t$ in this case. This part can also be implemented in $O(D+\sqrt{n})$ rounds, as follows. First, each tree edge learns how many times it is covered, which is an aggregate function of the non-tree edges that cover it and hence can be implemented in $O(D+\sqrt{n})$ rounds. If an edge $t \in R_k$ learns that it is covered 3 times, it would like to remove the higher petal of the global anchor $t_2$ below it. Note that all the edges learn all the global anchors and their petals in the algorithm, and using the LCA labels each edge can learn which of the higher petals cover it, and which of the anchors are below it which allows identifying the anchor $t_2$. Since there are $O(\sqrt{n})$ global anchors, all the vertices in the graph can now learn in $O(D+\sqrt{n})$ rounds the identity of all the global anchors that should be removed from $Y$, which completes the cleaning phase. 

The above gives a $(4+\epsilon)$-approximation for weighted TAP and a $(5+\epsilon)$-approximation for weighted 2-ECSS for any constant $\epsilon>0$ in $O((D+\sqrt{n})\log^2{n})$ rounds, as follows.

\begin{theorem}
There is a deterministic $(4+\epsilon)$-approximation algorithm for weighted TAP in the $\mathsf{CONGEST}$ model that takes $O((D+\sqrt{n})\frac{\log^2{n}}{\epsilon})$ rounds.
\end{theorem}  

\begin{proof}
The proof of Lemma \ref{approx_lemma} shows that if in the reverse-delete phase we cover all the edges with $y(t)>0$ at most 2 times, then we get a $(2+\epsilon)$-approximation for weighted TAP in the virtual graph $G'$. This gives a $(4+\epsilon)$-approximation for weighted TAP in the original graph by Lemma \ref{virtual_lemma}. The time complexity of the forward phase is $O((D+\sqrt{n})\frac{\log^2{n}}{\epsilon})$ rounds, and the reverse-delete phase has $O(\log^2{n})$ iterations that take $O(D+\sqrt{n})$ rounds. Other computations in the algorithm such as building a virtual graph, computing the layers and the decomposition also take $O((D+\sqrt{n})\log{n})$ rounds as discussed in the relevant sections. 
\end{proof}

As explained in Claim \ref{claim_TAP_2ECSS}, a $(4+\epsilon)$-approximation for weighted TAP gives a $(5+\epsilon)$-approximation for weighted 2-ECSS, by building an MST and then augmenting it to be 2-edge-connected, proving the following.

\twoECSS*

\remove{
\begin{theorem}
There is a deterministic $(5+\epsilon)$-approximation algorithm for weighted 2-ECSS in the CONGEST model that takes $O((D+\sqrt{n})\frac{\log^2{n}}{\epsilon})$ rounds.
\end{theorem} 
}

\newpage
\section{A Faster $O(\log n)$ Approximation via Low-Congestion Shortcuts}
\label{app:secondAlgo}
In this section, we explain how to obtain an $O(\log n)$-approximation algorithm for 2-ECSS that runs in $\mathsf{SC}(G)$ rounds in any network $G$. Here, $\mathsf{SC}(G)$
denotes the shortcut complexity of $G$~\cite{ghaffari2016shortcuts}, which is set equal to the minimum value $\alpha+\beta+\gamma$ such that for any partition of $G$ into vertex-disjoint connected parts, we have an algorithm that builds an $\alpha$-congestion $\beta$-dilation shortcut, in $\gamma$ rounds. We start with the high-level outline of the approach, presented in Section~\ref{subsec:log-outline}. Then, we fill out the details on how to perform each of the steps, by first providing some basic algorithmic tools in Section~\ref{subsec:log-tools}, and then we use these tools to build the desired subroutines for the outline, in Section~\ref{subsec:log-subroutines}.


\subsection{High-Level Outline}
\label{subsec:log-outline}

\paragraph{2-ECSS Approximation Boils Down to Tree Augmentation:} Let $T=(V, E_{T})$ be the minimum-weight spanning tree of $G=(V, E)$. As discussed in Claim \ref{claim_TAP_2ECSS}, to obtain an $O(\log n)$-approximate  minimum weight $2$-edge-connected spanning subgraph (2-ECSS), we would like to compute an $O(\log n)$-approximate minimum-weight set $S\subseteq E\setminus E_{T}$ that covers all edges of $E_T$. We refer to this set $S$ as the augmentation. As mentioned in Section \ref{sec:FirstAlgorithm-Overview}, this tree augmentation problem is a special case of the well-studied \emph{minimum-cost set cover} problem: in our case, we are choosing a minimum-weight collection of non-tree edges $S\subseteq E\setminus E_{T}$ which cover all the tree edges $E_T$.

Of course, in the distributed setting, this special case comes with its own difficulties, because the sets (in our case, non-tree edges) are not directly connected to their elements (in our case, tree edges). This means we cannot resort to standard distributed $O(\log n)$-approximations of the set cover problem, in a naive way. However, the high-level outline of our algorithm will still make use of a simple and well-known technique for parallelizing the greedy algorithm for minimum-cost set cover\cite{berger1989efficient}, similiar ideas were also used in \cite{DBLP:conf/podc/Dory18}. We explain an outline of the algorithm next. We refer the reader to \cite{berger1989efficient} for the details. Our focus will be on how to perform this outline for the tree augmentation problem, in $\mathsf{SC}(G)$ rounds.

\paragraph{Outline of Our Algorithm for Tree Augmentation:} We build the $O(\log n)$-approximate augmentation/cover set $S$ incrementally, in $O(\log^3 n)$ iterations. We next sketch how these iterations work. Our focus will be on how to perform each iteration in $\mathsf{SC}(G)$ rounds, using subroutines that we build on top of the known tools developed in the low-congestion shortcuts framework\cite{ghaffari2016shortcuts}.

The general idea is to mimic the standard sequential greedy approach, but in only $\poly(\log n)$ independent iterations. Consider an intermediate step of the algorithm. Suppose that some non-tree edges are already chosen to be in the output augmentation/cover. First, we mark each edge $e\in E_T$ if it is still not covered. We will explain how to perform this in $\mathsf{SC}(G)$ rounds, in Lemma~\ref{lem:covered-or-not}. For each non-tree edge $e'\in E\setminus E_{T}$, we define its \emph{cost-effectiveness} to be $\rho(e')=\frac{cover(e')}{weight(e')}$ where $cover(e')$ indicated the number of marked tree edges that $e'$ covers. We will explain how to compute this simultaneously for all non-tree edges, in $\mathsf{SC}(G)$ rounds, in Lemma~\ref{lem:how-many-covered}. Let $\Delta$ denote the maximum cost-effectiveness $\Delta:=\max_{e'\in E\setminus E_{T}}{\rho(e')}$. The sequential greedy algorithm for min-cost set cover would add one non-tree edge with cost-effectiveness $\Delta$, and repeat doing this, each time according to the new maximum cost-effectiveness. Of course, this sequential process is too slow for us.

To make the above more parallel, we would like to add many non-tree edges simultaneously, each with a cost-effectiveness at least $(1-\eps)\Delta$, for a small constant $\eps>0$. But there might be double counting in our cost-effectiveness calculations. That is, that many non-tree edges might be covering the same few tree-edges; in that case their union is not as cost-effective as their individual cost-effectiveness would suggest. To counter this effect, we use a the following definition. We call a subset $S' \subseteq (E\setminus E_{T})$ \emph{good} if the number of the unmarked edges that they cover is at least a $\Delta/100$ factor of their total cost $\sum_{e'\in S'} weight(e')$. Our algorithm will iteratively find good subsets of non-tree edges, with respect to the remaining uncovered tree edges, and add them to the cover $S$ that we are building. This will be continued until there is no tree edge that remains uncovered. By mimicking the analysis of the sequential greedy algorithm, as done in \cite{berger1989efficient}, it can be shown that any such algorithm which only adds good subset achieves an $O(\log n)$-approximation. Next, we discuss how we do this process in only $O(\log^3 n)$ iterations of finding good subsets.

\paragraph{Iterations of Tree Augmentation:} Our $O(\log^3 n)$ iterations are structured as $O(\log n)$ \emph{phases}, each made of $O(\log n)$ \emph{sub-phases}, where each sub-phase is simply $O(\log n)$ repetitions of the same procedure. Our $O(\log n)$ phases target values of cost-effectiveness that are decreasing powers $(1+\eps)^{U}$, $(1+\eps)^{U-1}$, \dots, $(1+\eps)^{L}$ of $1+\eps$. Here, $M=\log_{1+\eps} n$ and $L=\log_{1+\eps} 1/W_{max}$, and $W_{max}$ denotes the maximum weight of any edge. Since we assume $W_{max}\leq \poly(n)$, we have $M-L = O(\log n)$. 

Let us zoom into one phase. Let $\Delta=(1+\eps)^i$ be the current maximum cost-effectiveness, and let $A$ be the set of non-tree edges whose cost-effectiveness is at least a $\Delta(1-\eps)$. Notice that during this phase, as we add more and more non-tree edges to our cover, some edges of $A$ might have their cost-effectiveness drop below $\Delta(1-\eps)$, in which case they get removed from $A$. Now, we describe the sub-phases of this phase. These sub-phases are parametrized by the maximum number of non-tree edges in $A$ that cover a given tree edge $e$, and we will have $O(\log_{1+\eps} n)$ sub-phases, corresponding to decreasing powers $(1+\eps)^j$. 
Let us focus on one subphase. Let $d=(1+\eps)^j$ be the maximum number of non-tree edges in $A$ that cover a given uncovered tree edge $e$. Let $B$ be the set of all uncovered tree edges which are covered by at least $d(1-\eps)$ edges of $A$.  The sub-phase is made of $O(\log n)$ repetitions of the same random procedure, which attempts to find a good subset, to add to the tree augmenting set, as we describe next.

In each repetition, we sample each nontree edge of $A$ with probability $p=1/(2d)$. Then, we add the sampled set to our augmentation if it is a good set, i.e., it its cost-effectiveness is at least $\Delta/100$. Notice that we can easily check whether the set is good or not, in $O(D) + \mathsf{SC}(G)$ rounds, by gathering how many new tree edges are covered by the sampled non-tree edges, as well as the total weight of the sampled non-tree edges.
It can be seen that the probability that each remaining tree edge $e\in B$ is covered by our sampled set is at least $0.01$. Hence, after $O(\log n)$ repetitions, each tree edge in $B$ that has at least $d(1-\eps)$ neighbors in $A$ gets covered with high probability. Then, we proceed to the next subphase, where $d$ is lowered to $(1+\eps)^{j-1}$.

\paragraph{Remaining Components in the Above Outline} What remains from the above outline is the subroutines needed for performing each iteration. This is actually the only contribution of this paper, in this section. We will explain how we determine the covered tree edges and how we compute the number of marked tree edges that are covered by each non-tree edge (so that it can knows its cost-effectiveness). We can perform each of these two operations in $\widetilde{O}(\mathsf{SC}(G))$ rounds. Next, we first review some basic tools from the low-congestion shortcuts framework of \cite{ghaffari2016shortcuts} and then explain how to use these tools to perform our desired computations, in $\widetilde{O}(\mathsf{SC}(G))$ rounds.

\subsection{Basic Tools}
\label{subsec:log-tools}
In this section, we explain three basic tools that will be used in our algorithm. These tools allow us to compute descendants' sum, ancestors' sum, and heavy-light decompositions, for any tree $T$, in $\widetilde{O}(\mathsf{SC}(G))$ rounds. In terms of contributions, we borrow the Descendants' Sum tool --- stated in Theorem~\ref{thm:descendants}---from the work of Ghaffari and Haeupler~\cite{ghaffari2016shortcuts}, and we also provide new algorithms for the Ancestor's Sum problem (which needs to be solved somewhat differently), and also for heavy-light decompositions, respectively in Theorems~\ref{thm:ancestors} and \ref{thm:Heavy-Light}. Given the versatility of these basic primitives and their wide use when doing computations related to a tree structure, we are hopeful that they will find applications beyond our usage of them for the 2-ECSS problem.

\begin{definition}[\textbf{Descendants' Sum Problem}] 
We are given an arbitrary spanning tree $T=(V, E_T)$ and one input value $x(v)$ for each vertex $v\in V$. The objective is that each vertex $u\in V$ learns the summations $\sum_{v\in T_{u}} x(v)$ where $T_{u}$ denotes the set consisting of all the descendants of $u$, including $u$ itself. The summation operand can be replaced with any other aggregate function, e.g., maximum, minimum, etc.
\end{definition}
\begin{theorem}[Ghaffari and Haeupler \cite{ghaffari2016shortcuts}]
\label{thm:descendants} There is a randomized distributed algorithm that solves the Descendants' Sum Problem in $\mathsf{SC}(G)$ rounds.
\end{theorem}

\begin{definition}[\textbf{Ancestors' Sum Problem}] 
We are given an arbitrary spanning tree $T=(V, E_T)$ and one input value $x(v)$ for each vertex $v\in V$. The objective is that each vertex $u\in V$ learns the summations $\sum_{v\in P_{u}} x(v)$ where $P_{u}$ denotes the set consisting of all ancestors of $u$, including $u$ itself. The summation operand can be replaced with any other aggregate function, e.g., maximum, minimum, etc.
\end{definition}
\begin{theorem}\label{thm:ancestors} There is a randomized distributed algorithm that solves the Ancestors' Sum Problem in $\mathsf{SC}(G)$ rounds.
\end{theorem}
\begin{proof} We first invoke an algorithm of Ghaffari and Haeupler \cite{ghaffari2016shortcuts} which achieves two properties, in $\widetilde{O}(\mathsf{SC}(G))$ rounds: (1) it orients the tree such that each vertex knows its parent. (2) it decomposes the tree into an $L=O(\log n)$-level \emph{hierarchical partitioning}, with the following structure: (A) for each $i$, fragments of level-$i$ form a partitioning of the vertices of the tree and each level-$i$ fragment is a connected subgraph of the tree, (B) each level-$i$ fragment is made of merging one level-$(i-1)$ fragment with its level-$(i-1)$ children fragments, (C) at the very top, the level-$L$ fragment is simply the whole tree, and at the very bottom, each level-$0$ fragment is simply one vertex.

Given the above structure, we solve the ancestor' sum problem in a recursive manner. Let us consider the top of the hierarchy. The level-$L$ fragment is obtained by merging a level $L-1$ fragment $F_0$ with its children level $(L-1)$-fragments $F_1$, \dots, $F_d$. The ancestors' sum problem can be solved in $F_0$ with no dependency on $F_1$ to $F_d$, as the latter are only descendants of vertices of $F_0$. However, for a vertex $v\in F_i$ for $i\in [1, d]$ to have its ancestor's sum, the solution depends on $F_0$. In particular, let $w$ be the lowest ancestor of $v$ in fragment $F_0$, and let $w'$ be the highest ancestor of $v$ in fragment $F_i$. Then, the ancestor's sum of $v$ in the whole tree can be obtained by recursively solving the ancestor's sum of $v$ in fragment $F_i$ --- as it we remove the edge $\{w, w'\}$ --- and then adding to it the value $z$ which is equal to the ancestor's sum of $w$ in $F_0$. This gives us the recursive structure. Let $T(i)$ be the time to solve the ancestor's sum problem in a level-$i$ fragment. Also, let $U(i)$ denote the time to add a given fixed value $z$ to all the ancestor's sums in a level-$i$ fragment. Then, we can write $T(L)=T(L-1)+U(L-1)$. This is because to solve ancestor's sum on a level-$L$ fragment, it suffices to solve the problem on each of its level-$(L-1)$ fragments independently --- which takes $T(L-1)$ time --- and then increase each of the values in each of the children level-$(L-1)$ fragments by the value of the ancestor's sum of their lowest descendant in $F_0$.

The only thing that remains to solve this recursion is what is the complexity $U(i)$. This is simply the problem of delivering one fixed value the root of each level-$(i)$ fragment to all the vertices of this fragment. This problem can be solved directly using low-congestion shortcuts in $\widetilde{O}(\mathsf{SC}(G))$ rounds\cite{ghaffari2016shortcuts}. Recall that low-congestion shortcuts are exactly designed so that one can deliver one message (or an aggregate function such as minimum of many messages) to all vertices of the part/fragment. Since  $U(i)=\widetilde{O}(\mathsf{SC}(G))$ and $T(L)=T(L-1)+U(L-1)$, we get that $T(L) = \widetilde{O}(\mathsf{SC}(G)) \cdot L = \widetilde{O}(\mathsf{SC}(G)).$
\end{proof}
\begin{definition}[\textbf{Heavy-Light Decomposition for Trees}]
Given a tree $T=(V, E_T)$ and a root vertex $r\in V$, a heavy-light decomposition of $T$ is defined as follows. Let $T_v$ denote the subtree rooted in vertex $v\in V$ and let $|T_v|$ be the number of vertices in $T_v$, including $v$ itself. For each edge $e=\{v, u\}$ which connects a vertex $u$ to its parent $v$, we call $e$ a heavy edge if $|T_u|> T_{v}/2$ and light otherwise. Notice that any leaf-to-root path has at most $\log n$ light edges. Moreover, each vertex has at most one heavy edge to its children. The subgraph formed by heavy edges is a collection of disjoint paths, called heavy paths, each being from one vertex to one of its descendants. 

On the distributed side, we would like that each vertex $v$ learns the following information: (A) the number of vertices on the path $|P_v|$ connecting $v$ to $r$, (B) the identifier of each of the light edges $L_v$ in its path to the root $r$, where the identifier of a light edge $\{v, u\}$ is the identifiers of $v$ and $u$ as well as $|P_v|$ and $|P_{u}|$. 
\end{definition}

\begin{theorem}\label{thm:Heavy-Light}There is a randomized distributed algorithm that computes a heavy-light decomposition in $\widetilde{O}(\mathsf{SC}(G))$ rounds. Moreover, each two vertices $v, u$  which are adjacent in graph $G$ can know their Lowest Common Ancestor (LCA). 
\end{theorem}
\begin{proof} We first perform the descendants' sum operation starting with $x(v)=1$ for all vertices $v\in V$, so that each vertex $u$ knows $|T_u|$. This can be done in  $\widetilde{O}(\alpha+\beta+\gamma)$ rounds by Theorem~\ref{thm:descendants}. Then, heavy and light edges can be determined easily in one round, by each vertex $u$ sending its $T_u$ value to its parent $v$ and then the parent determining whether  $|T_u|> T_{v}/2$ or not.

Then, we can make each vertex learn the identifier of each of the light edges in its path to the root. For that, first, we perform an ancestors’ sum operation starting with initial values $x(v)=1 \forall v\in V$ so that each vertex $u$ knows $|P_u|$. Then, we using another ancestor's sum operation as follows: Define initial values $x(v)=\{(v, u, |P_{v}|, |P_{u}|)\}$ if the edge connecting $v$ to its parent $u$ is a light edge and $x(v)=\{\}$ otherwise. Define the summation operand on $x(v)$ to be the union operation. Notice that since each vertex-to-root path has at most $\log n$ light edges, the size of this union will never exceed $\log n$ tuples, when we perform union operations as we go down the tree. Then, by applying Theorem~\ref{thm:ancestors}, we can make each vertex $w$ know all of the light edges $L_w$ in its path to the root. 

Now, consider two vertices $v, u$ who are adjacent in $G$ and suppose that they exchange sets $L_v$ and $L_u$. Let $w$ be the lowest common ancestor of $v$ and $u$. Notice that $L_v\cap L_u = L_w$. Let $e=(y, y', |P_y|, |P_y'|)$ be the lowest light edge in $L_v \cap L_u$, i.e., the edge for which maximized $|P_{y}|$. Let $e_v=(y_v, y'_v)$ and $e_u=(z_v, z'_v))$ be the topmost light edges in $L_v$ and $L_u$ that are below $e$. Then, LCA $w$ is the vertex in $x\in \{y'_v, z'_v\}$ who has a smaller path length $|P_x|$ to the root $r$. 
\end{proof}
\subsection{Key Subroutines}
\label{subsec:log-subroutines}
Having discussed the above tools, we are now ready to go back to the outline of the tree augmentation, and discuss the subroutines that we would like to build.
As mentioned before, these two subroutines allow us to determine the covered tree edges and to compute the number of marked tree edges that are covered by each non-tree edge. We build these subroutines via applications of the tools discussed in the previous subsection. We provide these subroutines, in the following two lemmas.

\begin{lemma}\label{lem:covered-or-not} Suppose we are given a spanning tree $T=(V, E_T)$ of our graph $G=(V, E)$ and a set $S\subseteq E\setminus E_{T}$. There is a randomized distributed algorithm that for each edge $e\in E_T$, determine whether $e$ is covered by $S$ or not, simultaneously for all edges of $E_T$, in $\widetilde{O}(\mathsf{SC}(G))$ rounds.
\end{lemma}
\begin{proof}
For each edge $e\in E$, define its random identifier $RID(e)$ to be a random $10\log n$-bit string. This can be chosen by the higher identifier vertex of the edge and communicated to the lower identifier endpoint, in a single round. Notice that, with high probability, any two different edges have different random identifiers. 

For each vertex $v$, define $x(v)$ to be the XOR of the random identifiers of all of the edges in $S$ which are incident on $v$. If there is no such edge incident on $S$, define $x(v)$ to be a binary vector of length $10\log n$, where all the bits are zero. Then, perform a descendant's sum using Theorem \ref{thm:ancestors} so that each vertex $u\in V$ knows the summation $\oplus_{v\in T_u} x(v)$, where $\oplus$ means the bit-wise XOR. Notice that edges of $S$ whose both endpoint is in $T_u$ get canceled out in this XOR, because they are added to the XOR sum $\oplus_{v\in T_u} x(v)$ twice, once from each endpoint. Now, for each edge $e\in E_T$ which connects a vertex $u$ to its parent $w$, it suffices to check $\oplus_{v\in T_u} x(v)$. We say $e$ is covered iff this XOR is not a vector of all zeros. Notice that if $e$ is not covered, then deterministically the XOR is a vector of zeros. However, if $e$ is covered by some set of edges, then the XOR of those edges is with high probability not a string of all zeros, simply because each bit of that XOR is nonzero with probability $1/2$ and different bits are independent.
\end{proof}

\begin{lemma}\label{lem:how-many-covered} Suppose we are given a spanning tree $T=(V, E_T)$ of our graph $G=(V, E)$ and some of edges in $E_T$ are marked. There is a randomized distributed algorithm that for edge $e'\in E\setminus E_{T}$ determines how many unmarked edges $e\in E_T$ it covers, in $\widetilde{O}(\mathsf{SC}(G))$ rounds.
\end{lemma}
\begin{proof}
First, we make each vertex $v$ know the number $M_v$ of marked edges on its path $P_v$ to the root $v$.  This can be done using Theorem \ref{thm:ancestors} in $\widetilde{O}(\mathsf{SC }(G))$ time. Then, we make each vertex $v$ learn the \emph{modified identifier} of each of the light edges $L_v$ in its path to the root $r$, where we define the modified identifier of a light edge $\{v, u\}$ to be the following six entries: the identifiers of $v$ and $u$, numbers $|P_v|$ and $|P_{u}|$, and numbers $M_v$ and $M_u$. Notice that we can each each vertex learn the modified identifiers of all the light edges on its path to the root, using Theorem~\ref{thm:descendants} $\widetilde{O}(\mathsf{SC }(G))$, similar to what we did in Theorem~\ref{thm:Heavy-Light} for learning light edge identifiers on this path.

Then, for each non-tree edge $e'=(v, u)\in E\setminus E_{T}$, we first identify the LCA $w$ of $v$ and $u$, as explained in Theorem~\ref{thm:Heavy-Light}. Then, the number of marked edges covered by $e'$ is simply $M_v+M_u-2M_w$. Notice that vertices $v$ and $u$ know $w$ and $M_w$, as $w$ is one of the endpoints of one of the light edges in $L_v\cup L_u$, which means the information $M_w$ is tagged on the modifier identifier of that edge and is thus available. Once we have identified the LCA vertex $w$, as explained in Theorem~\ref{thm:Heavy-Light}, we can also determine $M_w$ by checking the information on one of the light edges connecting $w$ to one of its children.
\end{proof}

\section{Discussion}

In this paper, we show an algorithm for weighted 2-ECSS which is near-optimal both in terms of time complexity and in the approximation ratio obtained. The techniques we use here do not seem to allow improving the approximation below $(5+\epsilon)$. However, it would be interesting to study whether different approaches can get an approximation closer to $2$ in near-optimal time.

A key ingredient that allowed us obtaining a constant approximation in near-optimal time is a parallel algorithm for set cover problems with small neighborhood covers \cite{agarwal2018set}. In \cite{agarwal2018set}, the authors show many examples for additional problems having this property such as interval cover and bag cover. Therefore, we hope that the ideas we use here may be useful also for obtaining additional efficient distributed algorithms that give constant approximations for covering problems that have this structure.

This paper focuses on the weighted 2-ECSS problem. For $k>2$ the weighted $k$-ECSS problem seems to be much more involved, and currently the best complexity for approximating it is at least linear \cite{DBLP:conf/podc/Dory18}.
An interesting question is whether \emph{sublinear} algorithms exist for weighted $k$-ECSS also for larger values of $k$.\\

\textbf{Acknowledgements:} We thank Merav Parter for bringing \cite{kECSS} to our attention. This project has received funding from the European Union's Horizon 2020 Research And Innovation Program under grant agreement no.\ 755839. Supported in part by the Israel Science Foundation (grant 1696/14) and by the Swiss National Foundation (SNF) under project number $200021\_184735$. 

\bibliographystyle{plain}
\bibliography{2-ECSS}
 

\end{document}